\documentclass[preprint,5p,times,twocolumn]{elsarticle}

\usepackage{amssymb}
\usepackage{amsthm}
\usepackage{amsmath}
\allowdisplaybreaks

\usepackage{lineno}

\journal{J.\ Theor.\ Biol.}

\newtheorem{theorem}{Theorem}[section]

\newcommand{\pp}[2]{\frac{\partial {#1}}{\partial {#2}}}
\newcommand{\dd}[2]{\frac{\mathrm{d}{#1}}{\mathrm{d}{#2}}}

\newcommand{\Vc}{V_{\mathrm{c}}}
\newcommand{\Ag}{A_{\mathrm{g}}}
\newcommand{\rg}{r_{\mathrm{g}}}
\newcommand{\rgtilde}{\tilde{r}_{\mathrm{g}}}
\newcommand{\sg}{s_{\mathrm{g}}}
\newcommand{\sgtilde}{\tilde{s}_{\mathrm{g}}}
\newcommand{\lc}{l_{\mathrm{c}}}
\newcommand{\cb}{\bar{c}}
\newcommand{\Cb}{\bar{C}}
\newcommand{\lb}{\bar{l}}
\newcommand{\cc}{c_{\mathrm{c}}}
\newcommand{\Cc}{C_{\mathrm{c}}}

\newcommand{\cs}{c_{\mathrm{s}}}
\newcommand{\ccb}{\bar{c}_{\mathrm{c}}}

\newcommand{\Sa}{S_{\!\mathrm{a}}}

\newcommand{\cinf}{c_{\infty}}
\newcommand{\linf}{l_{\infty}}
\newcommand{\cmax}{c_{\mathrm{max}}}

\newcommand{\lmin}{l_{\mathrm{min}}}
\newcommand{\linfone}{l_{\infty 1}}
\newcommand{\linftwo}{l_{\infty 2}}

\newcommand{\rme}{\mathrm{e}}

\begin{document}

\begin{frontmatter}



\title{A one-dimensional moving-boundary model
for tubulin-driven axonal growth}


\author{S.~Diehl, E.~Henningsson, A.~Heyden and S.~Perna}
\address{Centre for Mathematical Sciences, Lund University, P.O.~Box 118,
S-221~00 Lund, Sweden. E-mail: {\tt diehl@maths.lth.se}, {\tt
erikh@maths.lth.se}, {\tt heyden@maths.lth.se}, {\tt
stefano.perna@studenti.unimi.it}}

\begin{abstract}
A one-dimensional continuum-mechanical model of axonal elongation due to
assembly of tubulin dimers in the growth cone is presented. The conservation of
mass leads to a coupled system of three differential equations. A partial
differential equation models the dynamic and spatial behaviour of the
concentration of tubulin that is transported along the axon from the soma to
the growth cone. Two ordinary differential equations describe the
time-variation of the concentration of free tubulin in the growth cone and the
speed of elongation, respectively. All steady-state solutions of the model are
categorized. Given a set of the biological parameter values, it is shown how
one easily can infer whether there exist zero, one or two steady-state
solutions and directly determine the possible steady-state lengths of the axon.
Explicit expressions are given for each stationary concentration distribution.
It is thereby easy to examine the influence of each biological parameter on a
steady state. Numerical simulations indicate that when there exist two steady
states, the one with shorter axon length is unstable and the longer is stable.
Another result is that, for nominal parameter values extracted from literature,
in a large portion of a fully grown axon the concentration of free tubulin is
lower than both concentrations in the soma and in the growth cone.
\end{abstract}


\end{frontmatter}


\section{Introduction}\label{sect:introduction}

Axons are cables that transmit electrical signals between neurons. When the
cell body of a neuron, the soma, is fully formed, the neuron begins to sprout
small projections known as neurites. After an initial stage and through a
process not fully detailed yet, one of these neurites exhibits a dramatic
increase in growth and is denoted the axon. The axon elongates and seeks its
target in the body. This stage is mainly guided by the mobile and sensitive tip
of the axon called the growth cone, which is highly responsive to chemical
substances in the body environment. Such {chemical cues} can attract or repel
the growth cone. This is achieved by a reorganization of the internal protein
structure, the cytoskeleton, inside the growth cone. If a gradient of guidance
cue is found, cytoskeletal changes happen asymmetrically and the growth cone
turns towards or away from the guidance cue. A description and review of the
cytoskeletal dynamics and transport in growth cone motility and axon guidance
is provided by Dent and Gertler~\cite{Dent2003}. A later review by Suter and
Miller~\cite{Suter2011} focuses on the influence of the forces generated by the
growth cone on the axonal elongation. Excellent reviews of the different types
of modelling of different stages in the development of axons and their
behaviour are provided by Graham and van Ooyen~\cite{Graham2006BMC}, Kiddie et
al.~\cite{Kiddie2005} and van Ooyen~\cite{vanOoyen2011}. Additional biological
insight are provided by Miller and Heidemann~\cite{Miller2008}. In a recent
publication, Hjort et al.~\cite{Hjort2014} model the competitive tubulin-driven
outgrowth and withdrawal of different branches of the same neuron.

The elongation of the axon is caused by an assembly (polymerization) of free
tubulin dimers to microtubules that build up the cytoskeleton. This occurs
mainly in the growth cone, while tubulin is produced in the soma. The fact that
axons can grow very long has initiated both experimental and theoretical
investigations of the biological and physical processes which are responsible
for the transportation of tubulin along the long axon. Fundamental variables
that affect the growth are the amount of available free tubulin in the growth
cone, the reaction rates of the polymerization and depolymerization processes
in the growth cone, the degradation of tubulin in the entire axon, the tubulin
production rate in the soma and the processes of transportation of tubulin
along the axon.

We have been attracted by the question of whether it is possible to develop a
simple mechanistic model for axonal growth so that some investigations can be
performed on how the key biological parameters influence the axonal growth.
Hence, we will confine to a one-dimensional model and neither consider any
axonal pathfinding nor the branching of axons which sometimes occur. Axons may
not only grow but also contract or stay still during time periods, and also
alternate between these three phases. This has lead to include stochastic
variables in models. For example, Janulevicius et al.~\cite{Janul2006} model
the growth and shrinkage phases of microtubules and the random switches between
these phases. Model parameters are collected from the experimental work by
Walker et al.~\cite{Walker1988}. Deterministic models, which do not contain
such random processes, try to catch the mean value of the behaviours of many
microtubules.

All mathematical models of biological phenomena are formulated with substantial
simplifications. Some phenomenological models utilize mathematical functions to
describe certain connections; for example, O'Toole and Miller~\cite{OToole2011}
quantify axonal elongation in terms of slow axonal transport, protein
degradation, protein density, adhesion strength and axonal viscosity.

Models of the dynamic behaviour are formulated by differential equations. Some
dynamic models confine to ordinary differential equations (ODEs), hence
containing unknown functions (the model outputs) that depend only on time; see
e.g.\ \cite{vanOoyen2001,VanVeen1994}. Because of the substantial length of an
axon it is, however, natural to assume that the concentration of a substance
varies both with time and position along the axon. Deriving a dynamic model
from a physical law, such as the conservation of mass, leads then necessarily
to a partial differential equation (PDE); see
\cite{Garcia2012,Graham2006,McLean2004,McLean2006,McLean2004num,Sadegh2010,Smith2001}.
The multi-compartment model by Hjort et al.~\cite{Hjort2014}, which has a
system of ODEs for each axon, is essentially a spatially discretization of a
PDE. However, it is generally safer to write a model including functions of
several variables (spatial location and time) in terms of PDE and then use
established numerical methods for their simulation.

Garcia et al.~\cite{Garcia2012} model the diffusion of tubulin along the axon
with the linear diffusion PDE and the active transport by tracking the position
of each motor protein assuming they all move with a constant velocity. The
microtubule assembly at the tip is modelled by an ODE with a constant
polymerization rate and they ignore the depolymerization and degradation of
tubulin. Furthermore, they propose a mechanical model of the polymerized
microtubules in order to describe how the axonal growth process influences the
mechanical properties of the cytoskeleton.

Smith and Simmon~\cite{Smith2001} presented and analyzed an interesting model
of the two components of axonal cargo transport, diffusion and active, of a
substance by means of three linear PDEs, all originating from the conservation
of mass of one substance present in three states. One diffusion equation models
the free substance, and the other two advection equations model the anterograde
and retrograde moving cargoes (by motor proteins), respectively. The coupling
between the equations occur via source terms, or rather binding/detachment
terms, which model the movements of substance between the free state and either
of the actively moving-cargo states, in both directions. There are seven model
parameters; one diffusion constant, two advection velocities and four rate
constants for the bindings/detachments. The axon is assumed to have a fixed
length and there is no degradation of the substance included in the model.

The model by Smith and Simmon~\cite{Smith2001} has been used by Sadegh Zadeh
and Shah~\cite{Sadegh2010}, who successfully calibrated the parameters of the
model to published experimental data.

As noted by Smith and Simmon~\cite{Smith2001}, a simplified model of their
three linear PDEs consists of a single advection-diffusion PDE with only two
model parameters; an effective drift velocity and an effective drift diffusion
constant; see \cite[Formulas (4a)--(4b)]{Smith2001}. We are interested in such
a simplified approach for the transport of tubulin along the axon. Such an
advection-diffusion PDE, with an additional sink term modelling the degradation
of tubulin, can be found in the model by McLean and Graham~\cite{McLean2004}.

The present work is particulary motivated by the deterministic one-dimensional
continuum model by McLean and Graham~\cite{McLean2004}. The linear PDE
describes the concentration of tubulin as a function of time and distance from
the soma along the axon to the growth cone. The moving growth cone means that
the right boundary of the PDE is moving with an unknown velocity, which is the
axon growth velocity. Such a moving-boundary problem causes some mathematical
difficulties. The axonal elongation speed is modelled by an ODE, which should
capture the assembly and disassembly of microtubules depending on the available
free tubulin concentration at the moving boundary. The different steady-state
solutions are presented in \cite{McLean2004}. In
\cite{Graham2006,McLean2004num}, numerical simulations of the model are
presented. In \cite{Graham2006}, Graham et al.\ also introduced an additional
parameter, the autoregulation gain, related to the production of tubulin in the
soma. In \cite{McLean2006}, the interesting question of stability of the steady
states was investigated.

Our model can be seen as an extension and modification of the one by McLean and
Graham~\cite{McLean2004}. In addition to the linear PDE and ODE described
above, we have an ODE modelling the dynamics of the free tubulin concentration
in the growth cone where the motor proteins release the tubulin dimers. Another
difference is the modelling of the flux of tubulin over the moving boundary.

The rest of the paper is organized as follows. The model is derived in detail
in Section~\ref{sect:model} with the aim of motivating all parameters. The
nominal parameter values chosen are motivated in Section~\ref{sect:parameters}.
All steady-state solutions of the model are given and characterized in
Section~\ref{sect:SS}. Section~\ref{sect:plotsSSpar} contain plots of how the
different steady-state solutions depend on the variation of each parameter from
the nominal set of parameter values. Investigations of the stability of the
steady states by means of numerical simulations are given in
Section~\ref{sect:stability}. Discussions and conclusions, including a detailed
description of the difference between our model and the one by McLean and
Graham~\cite{McLean2004}, are given in Section~\ref{sect:discussion}.

\section{The Model}\label{sect:model}

In this section, a continuum mechanical model for the growth of an axon is
derived. Any modelling methodology for capturing a biological or physical
phenomenon has to start by making idealizing assumptions. In neuronal
physiology it is a known fact that the growth of a newborn axon is strictly
connected to the presence of the group of proteins called tubulin. The main
idealizing assumption is that tubulin is the only substance involved in the
growth of an axon. Another idealizing assumption is that the molecules of free
tubulin are so small that one can consider them as a homogeneous continuum.
Then the physical conservation law of mass can be used to derive differential
equations that govern the dynamic behaviour of the concentration of tubulin
both along the axon and in the growth cone. Finally, some constitutive
assumptions have to be introduced to couple and thereby reduce the number of
unknown variables so that a solution of the equations can be obtained. Such
constitutive assumptions are based on physical and biological facts as far as
possible.

\subsection{Idealizing modelling assumptions}

\begin{figure}
\begin{center}
\includegraphics[width=0.4\textwidth]{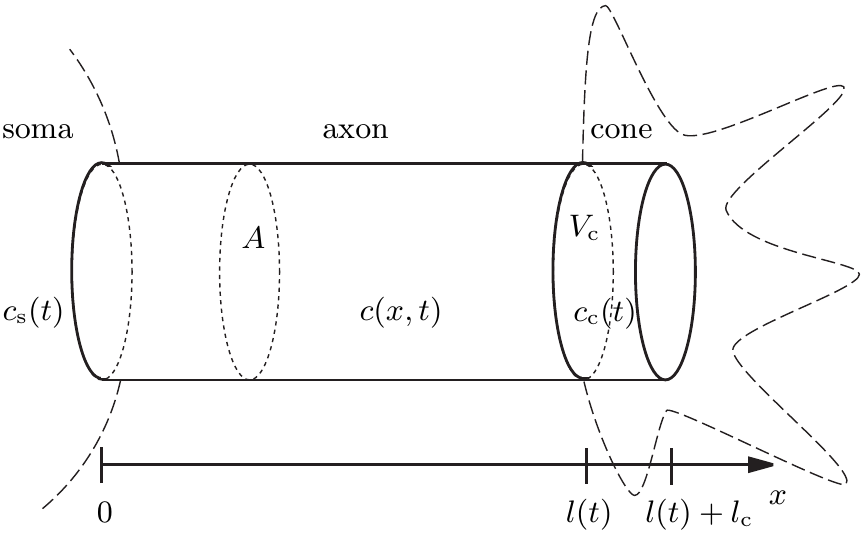}
\end{center}
\caption{Schematic illustration of a growing axon.}
\label{fig:Growing_axon}
\end{figure}%
In Figure~\ref{fig:Growing_axon}, an idealized axon is shown. The
one-dimen\-sional $x$-axis is placed along the axon, which at time $t$~[s] has
the length $l(t)$~[m]. The latter variable is one output of the model. The
effective cross-sectional area of the axon through which tubulin is transported
is assumed to be a constant denoted by $A$~[m$^2$]. Another output of the model
is the concentration of tubulin along the axon, which is assumed to vary with
$x$ and $t$; i.e.\ $c=c(x,t)$~[mol$/$m$^3$]. Tubulin is produced only in the
cell body, the {soma}, which is placed to the left of $x=0$. We assume that the
time-varying tubulin concentration in the soma $\cs(t)$ is known. No tubulin is
produced along the axon, but degradation occurs at the constant rate
$g$~[1$/$s]. The tip of the axon, the {growth cone}, is located to the right of
$x=l(t)$, and has the volume $\Vc$~[m$^3$]. For simplicity, we introduce a
characteristic length of the growth cone; $\lc:=\Vc/A$ (it turns out that the
final model only contains this ratio). The growth cone is considered to be a
completely mixed compartment in which the unknown concentration of tubulin
dimers is denoted by $\cc(t)$. No production of tubulin occurs in the cone, but
consumption occurs, partly because of degradation at the constant rate
$g$~[1$/$s], partly because of the assembly of dimers to microtubules that
elongates the axon at a constant rate $\rgtilde$~[1$/$s], i.e., $\rgtilde$ is
the reaction rate of polymerization of guanosine triphosphate (GTP) bound
tubulin dimers to microtubule bound guanosine diphosphate (GDP). As for the
polymerization, we let $\Ag$~[m$^2$] denote the constant effective area of
growth and $\rho$~[mol$/$m$^3$] the density of the assembled microtubules (the
cytoskeleton). Finally, we assume that the assembled microtubules in the growth
cone may disassemble at the constant rate $\sgtilde$~[$1/$s]. All biological
and physical constants are assumed to be non-negative.

\subsection{The conservation law of mass and constitutive assumptions}

The conservation of mass of tubulin states that the rate of increase of mass in
an arbitrary interval $(x_1,x_2)$ of the $x$-axis equals the mass flux (mol per
unit time) in minus the flux out plus the production inside the interval:
\begin{equation}\label{eq:conslawint}
\dd{}{t}\int\limits_{x_1}^{x_2} Ac(x,t)\,\mathrm{d}x
=A\left(F|_{x=x_1}-F|_{x=x_2}\right)+
\int\limits_{x_1}^{x_2}A\Sa\,\mathrm{d}x,
\end{equation}
where $F$~[mol$/($m$^2$s$)$] is the flux per unit area of tubulin and
$\Sa$~[mol$/($m$^3$s$)$] is a source/sink modelling local
production/consumption in the axon. Since only degradation is taken into
account, we have $\Sa=-gc$. Note that the flux $F$ is the product of the
concentration $c$ and the velocity of tubulin $v$~[m$/$s]:
\begin{equation}\label{eq:F=vc}
F=cv.
\end{equation}

The conservation law \eqref{eq:conslawint} yields two dynamic equations.
Letting the interval $(x_1,x_2)$ shrink to a point $x\in(0,l(t))$ one gets the
PDE (assuming $c$ and $F$ are continuously differentiable functions):
\begin{equation}\label{eq:PDE}
\pp{c}{t}+\pp{F}{x} = -gc\quad\text{for}\quad 0<x<l(t).
\end{equation}
This is one equation with two unknowns; $c$ and $F$ (or $v$). A constitutive
assumption is needed. It is often assumed that the flux of tubulin $F$ is
determined by active transport by motor proteins having the constant velocity
$a$~[m$/$s] and diffusion of free tubulin according to Fick's law
\citep{Galbraith2000,Galbraith1999,Graham2006BMC,Hjort2014,Kiddie2005,McLean2004,Miller2008,
Smith2001,vanOoyen2001,VanVeen1994,Sadegh2010}. The constitutive assumption
relating the flux $F$ and the concentration $c$ is
\begin{equation}\label{eq:fluxlaw}
F(c,c_x) = ac-Dc_x,
\end{equation}
where $D$~[m$^2/$s] is the diffusion coefficient and $c_x=\partial c/\partial
x$. Substitution of \eqref{eq:fluxlaw} into PDE \eqref{eq:PDE} gives the
advection-diffusion-reaction equation
\begin{equation}\label{eq:PDE1}
\pp{c}{t}+a\pp{c}{x}-D\frac{\partial^2c}{\partial x^2} = -gc\quad\text{for}\quad 0<x<l(t).
\end{equation}

The conservation of mass \eqref{eq:conslawint} should also hold for the growth
cone located in the interval $(x_1,x_2)=(l(t),l(t)+\lc)$. Over (the moving)
right boundary $x=l(t)+\lc$, there is no flux; however, there is over $x=l(t)$.
To write down this flux from left to right, we introduce the following
notation:
\begin{align*}
c^-&:=c\big(l(t)^-,t\big)=\lim_{\epsilon\searrow 0} c\big(l(t)-\epsilon,t\big),
\end{align*}
which is the concentration just to the left of the boundary $x=l(t)$.
Similarly, we use the notation $c_x^-$ and $v^-$. The boundary $x=l(t)$ moves
with velocity $l'(t)=\mathrm{d}{l}/\mathrm{d}{t}$ while the tubulin moves with
the velocity $v^-$. If $v^-$ is higher than $l'(t)$, then there is a net inflow
of tubulin to the growth cone; otherwise, there is an outflow. Hence, it is the
relative velocity of tubulin $v^--l'(t)$ that determines whether the net influx
is positive or negative. Since the flux (amount per unit time) is the product
of the net velocity, the concentration and the cross-sectional area, the net
flux over $x=l(t)$ in the $x$-direction is
\begin{align}\label{eq:temp2}
A\big(v^--l'(t)\big)c^-.
\end{align}
With the physical law \eqref{eq:F=vc} and the constitutive assumption for $F$
\eqref{eq:fluxlaw}, the velocity of tubulin can be written
\begin{align*}\label{eq:v}
&v(c,c_x) = \frac{F(c,c_x)}{c}=a-D\frac{c_x}{c}.
\end{align*}
Substituting this into \eqref{eq:temp2}, we get the flux [mol$/$s] of tubulin
over $x=l(t)$ (seen by an observer moving with the boundary):
\begin{equation}\label{eq:fluxoverl}
A\big(ac^--Dc_x^--l'(t)c^-\big).
\end{equation}
The conservation of the amount of free tubulin in the growth cone leads to the
following ODE:
\begin{equation}\label{eq:cons_law_cone}
\begin{split}
&\underbrace{\dd{(\Vc\cc)}{t}}_{\text{mass increase per unit time}}=\\
&=\underbrace{A\big(ac^--Dc_x^--l'(t)c^-\big)}_{\text{flux in}}
-\underbrace{g\Vc\cc}_{\text{degradation}}
-\underbrace{\rgtilde\Vc\cc}_{\text{assembly}}
+\underbrace{\sgtilde\rho\Ag\kappa\lc}_{\text{disassembly}}.
\end{split}
\end{equation}
We assume that the degradation of free tubulin occurs at the same rate $g$ as
in the entire axon. The two last terms describe mass per unit time of the
assembly/disassembly (polymerization/deplymerization) of microtubules. The
assembled mass per unit time is assumed to be proportional to the available
amount of tubulin in the cone $\Vc\cc(t)$ with the reaction rate $\rgtilde$ as
the proportionality constant. The disassembly at the rate $\sgtilde$ is a
process that reduces the growth rate, and is proportional to the already
assembled microtubules, $\rho\Ag\kappa\lc$, where $\kappa>0$ is a dimensionless
constant such that $\kappa\lc$ is the length of the assembled microtubles that
may undergo disassembly. Experimental evidence that this term is independent of
the free tubulin concentration has been provided by Walker et
al.~\cite{Walker1988}.

Before simplifying Equation~\eqref{eq:cons_law_cone}, we note that it contains
the two unknown functions $\cc(t)$ and $l(t)$. A relation between these is
needed. The conservation of the amount of assembled (polymerized) tubulin is
the following, where the terms on the right-hand side have been explained
above; see the two last terms of \eqref{eq:cons_law_cone}:
\begin{align}\label{eq:dldt}
&\underbrace{\dd{(\rho\Ag l)}{t}}_{\text{mass increase per unit time}}
=\underbrace{\rgtilde\Vc\cc(t)}_{\text{assembly}}
-\underbrace{\sgtilde\rho\Ag\kappa\lc}_{\text{disassembly}}.
\end{align}
Since the density $\rho$ and growth area $\Ag$ of assembled microtubules are
assumed to be constants, we can divide Equation~\eqref{eq:dldt} by $\rho\Ag$
and obtain
\begin{align}\label{eq:dldt1}
&\dd{l}{t}=\frac{\rgtilde\Vc}{\rho\Ag}\cc(t)-\sgtilde\kappa\lc.
\end{align}
If the concentration of tubulin is too low, this equation yields that the
elongation is negative, which means that the axon shrinks. To compare with the
notation in the model of McLean and Graham~\cite{McLean2004}, who write the
equation $l'(t)=\rg\cc-\sg$, we set
\begin{align*}
&\rg:=\frac{\rgtilde\Vc}{\rho\Ag}\quad\text{and}\quad\sg:=\sgtilde\kappa\lc,
\end{align*}
where $\sg$ is easily interpreted as the maximum speed of shrinkage, which
occurs when $\cc=0$. The parameter $\rg$ is a sort of rate coefficient for the
polymerization; however, it has not the unit of a rate constant. It is the
proportionality concentration between the resulting elongation speed due to the
polymerization and $\cc$ when there is now disassembly. We also define
\begin{align}\label{eq:cinf}
&\cinf:=\frac{\sg}{\rg}=\frac{\sgtilde\rho\Ag\kappa\lc}{\rgtilde\Vc}=\frac{\sgtilde\rho\Ag\kappa}{\rgtilde
A},
\end{align}
which is the threshold concentration when $l'(t)=0$, i.e.\ when there is no
elongation, since the processes of assembly and disassembly are equally large.
Therefore, $\cinf$ is the steady-state concentration in the growth cone. Then
\eqref{eq:dldt1} can be written
\begin{align}\label{eq:dldt2}
&\dd{l}{t}=\rg(\cc(t)-\cinf).
\end{align}
Thus, $\rg$ is the proportionality constant between the elongation speed and
the excess tubulin concentration in the growth cone above the steady-state
level $\cinf$.

We now simplify Equation~\eqref{eq:cons_law_cone} by dividing by $A$, use
\eqref{eq:cinf} to introduce $\cinf$ and use \eqref{eq:dldt2} to express it
without $l'(t)$:
\begin{equation}\label{eq:dccdt1}
\lc\dd{\cc}{t}=ac^--Dc_x^--\rg(\cc-\cinf)c^--g\lc\cc-\rgtilde\lc\cc+\rgtilde\lc\cinf.
\end{equation}

\subsection{Model equations with boundary and initial conditions}

The three model equations are \eqref{eq:PDE1}, \eqref{eq:dldt2} and
\eqref{eq:dccdt1}, and the unknowns are $c(x,t)$, $\cc(t)$ and $l(t)$. The
physical/biological parameters that need to be specified are $a$, $D$, $g$,
$\lc$, $\rgtilde$, $\rg$, $\cinf$ and $\cs(t)$. When diffusion is present
($D>0$) it is well known that the parabolic equation~\eqref{eq:PDE1} has smooth
solutions $c(x,t)$. Therefore, it is natural to impose boundary conditions
requiring that the concentration is continuous at $x=0$ and $x=l$. In
particular, this means that $c^-=c(l(t),t)=\cc(t)$, which simplifies the
right-hand side of \eqref{eq:dccdt1}. Initially, we assume that the neurite
that becomes the axon has the length $l_0>0$ (which may be small) and that the
given initial concentration distribution is $c_0(x)$ (the index denotes $t=0$).
From a mathematical point of view, the function $c_0(x)$ can be chosen rather
freely; however, a natural choice for a small neurite may be the soma
concentration. The full model is the following:
\begin{equation}\label{eq:model}
\left\{
\begin{aligned}
&\pp{c}{t}+a\pp{c}{x}-D\frac{\partial^2c}{\partial x^2} = -gc,&\quad&0<x<l(t),\ t>0,\\
&\begin{split}
&\lc\dd{\cc}{t}=(a-g\lc)\cc-Dc_x^-\\
&\qquad -(\rg\cc+\rgtilde\lc)(\cc-\cinf),
\end{split}&& t>0,\\
&\dd{l}{t}=\rg(\cc-\cinf),&&t>0,\\
&c(0,t) = \cs(t),&&t\ge 0,\\
&c\big(l(t),t\big)=\cc(t),&&t>0,\\
&c(x,0) =c_0(x),&&0\le x\le l(0)=l_0,\\
&\cc(0)=c_0(l_0).&&\\
\end{aligned}
\right.
\end{equation}
Since we are interested in steady-state solutions in this work, we will mostly
use a constant soma concentration.

\section{Parameter values}\label{sect:parameters}

The values of the model parameters are collected from the literature. The
nominal values and possible intervals are given in Table~\ref{table:1}. Here,
we motivate our choices.

\begin{table}[tbh]
\caption{Parameter values. In steady state, the nominal parameter values
together with the constant soma concentration $\cs=\cinf$ imply $\linf=65.6$~mm
(see Figure~\ref{fig:varyc0}).}\medskip

\begin{tabular}{llll}
Parameter & Nominal value & Interval & Unit \\
\hline
$a$ & $1$ & $0.5$--$3$ & $10^{-8}$~m$/$s \\
$D$ & $10$ & $1$--$25$ & $10^{-12}$~m$^2/$s \\
$g$ & $5$ & $1$--$200$ & $10^{-7}$~s$^{-1}$ \\
$\lc$ & $4$ & $0.1$--$1000$ & $10^{-6}$~m \\
$\rg$ & $1.783$ & --- & $10^{-5}$~m$^4/($mol\,s$)$ \\
$\rgtilde$ & 0.053 & --- & s$^{-1}$ \\
$\cinf$ & $11.90$ & --- & $10^{-3}$~mol$/$m$^3$ \\
$\cs$ & --- & $0$--$4\cinf$ & mol$/$m$^3$ \\
 \end{tabular}
\label{table:1}
\end{table}

Galbraith et al.~\cite{Galbraith1999} present experimentally determined values
for the transportation of tubulin in the giant squid axon. For the active
transport, they obtained $a=93.9~\mu$m$/$h $\approx 2.3$~mm$/$day $\approx
2.66\cdot 10^{-8}$~m$/$s and the diffusion coefficient was about
$D=8.59~\mu$m$^2/$s~$=8.59\cdot 10^{-12}$~m$^2/$s. Reported active transport
speeds for other animals are $a\approx$~0.5--2~mm$/$day; see Galbraith and
Gallant~\cite[Table~1]{Galbraith2000} and Miller and
Samuels~\cite[Table~1]{Miller1997}, and the references therein. Experiments
with cultured neuronal cells by Keith~\cite{Keith1987} yielded the two
different values $a=2.3$~mm$/$day and $0.31$~mm$/$day for the slow components
\emph{a} and \emph{b}, respectively.

Salmon et al.~\cite{Salmon1984} have measured the radial diffusion of tubulin
in cytoplasm of eggs and embryos of the sea urchin {\em Lytechinus variegatus}
and via the linear diffusion equation estimated the diffusion coefficient to
$D=5.9\cdot 10^{-12}$~m$^2/$s. Other experiments have resulted in diffusion
coefficients of the same order, e.g., Pepperkok et al.~\cite{Pepperkok1990}
report values between $1.3\cdot 10^{-12}$~m$^2/$s and $1.6\cdot
10^{-12}$~m$^2/$s.

As for the tubulin degradation, Caplow et al.~\cite{Caplow2002} report the
half-times of 9.6~h for tubulin-GTP and 2.4~h for tubulin-GDP, which give the
degradation rates $g=\ln 2/(9.6\cdot 3600)~\text{s}^{-1}\approx 2.0\cdot
10^{-5}~\text{s}^{-1}$ and $5\cdot 10^{-6}~\text{s}^{-1}$, respectively. Miller
and Samuels~\cite[Table~1]{Miller1997} report several published experimental
results with half-times of 14--75~days, which correspond to $g$ between
$1.1\cdot 10^{-7}~\text{s}^{-1}$ and $5.7\cdot 10^{-7}~\text{s}^{-1}$. It seems
to be of interest to investigate the outputs of the model for a large interval
of values of $g$.

As in the model by McLean and Graham~\cite{McLean2004}, we condense all
complicated building processes in the growth cone and assume that the axonal
elongation can be described by Equation~\eqref{eq:dldt2}. Appropriate values
should be found for the concentration-rate coefficient $\rg$ and either the
steady-state concentration $\cinf$ or the maximum speed of shrinkage
$\sg=\rg\cinf$. It is sometimes argued that the speed of axonal elongation from
the soma is similar, or equal, to the speed of growth $\upsilon_+$ of the
plus-end assembly of individual microtubules \cite{Graham2006,Mitchison1984a}.
Published experiments have shown that the growth speeds $\upsilon_+$ and
$\upsilon_-$ for the plus and minus end, respectively, of an independent
microtubule depend on the surrounding free tubulin concentration $\cc$
according to the affine relationships
\begin{align}
\upsilon_+&=\alpha_+\cc-\beta_+,\label{eq:v+}\\
\upsilon_-&=\alpha_-\cc-\beta_-,\label{eq:v-}
\end{align}
for positive constants $\alpha_+>\alpha_-$, $\beta_+$ and $\beta_-$
\cite{Mitchison1984a,Walker1988}. The dynamic growth of a microtubule consists
of phases of growth alternated by rapid shortening due to depolymerization. The
shortening phase has a speed that seems to be independent of concentration
\cite{Walker1988}. Furthermore, both the growth and shortening phases contain
periods of pauses.

Mitchison and Kirschner~\cite{Mitchison1984a} report the experimental values
\begin{align*}
&\text{$\alpha_+=0.135~$m$/(\text{M\,min})=2.25\cdot 10^{-6}~\text{m}^4/(\text{mol\,s})$}, \\
&\text{$\alpha_-=0.042~$m$/(\text{M\,min})=0.70\cdot 10^{-6}~\text{m}^4/(\text{mol\,s})$},
\end{align*}
and $\beta_+$ and $\beta_-$ small (near zero). For the same tubulin
preparation, the steady-state concentration, when the average
elongation/shortening of many microtubules is zero, was $14~\mu$M~$=14\cdot
10^{-3}$~mol$/$m$^3$, which together with the obtained constants $\alpha_{\pm}$
means that both ends are growing in steady-state ($\upsilon_+,\upsilon_->0$).
Hence, there are other processes, such as the rapid shortening, that make the
average growing speed be zero. Mitchison and Kirschner~\cite{Mitchison1984a}
argue that in steady-state the majority of the microtubules that grow slowly is
balanced by the minority shrinking rapidly.

Graham et al.~\cite{Graham2006} refer to the value $\alpha_+=2.25\cdot
10^{-6}~\text{m}^4/(\text{mol\,s})$ of Mitchison and
Kirschner~\cite{Mitchison1984a} when they choose $\rg=2.78\cdot
10^{-6}~\text{m}^4/(\text{mol\,s})$ in Equation~\eqref{eq:dldt2} for their
model, which is the one by McLean and Graham~\cite{McLean2004}. Then they
choose $\cinf=10~\mu$M as the common order of concentration.

Walker et al.~\cite{Walker1988} presented further detailed measurements of the
individual microtubules in porcine brain tubulin and obtained the following
values for \eqref{eq:v+}--\eqref{eq:v-}:
\begin{align}
&\text{$\alpha_+=0.33~$m$/(\text{M\,min})=5.5\cdot 10^{-6}~\text{m}^4/(\text{mol\,s})$},\label{eq:a+Walker} \\
&\text{$\beta_+=1.59\cdot 10^{-6}$~m$/\text{min}=2.65\cdot 10^{-8}$~m$/\text{s}$},\label{eq:b+Walker} \\
&\text{$\alpha_-=0.15~$m$/(\text{M\,min})=2.5\cdot 10^{-6}~\text{m}^4/(\text{mol\,s})$},\notag \\
&\text{$\beta_-=0.85\cdot 10^{-6}$~m$/\text{min}=1.42\cdot 10^{-8}$~m$/\text{s}$}.\notag
\end{align}
It is interesting to note that, with these values, the two lines
\eqref{eq:v+}--\eqref{eq:v-} intersect at the concentration $\approx 5~\mu$M,
which is also the concentration at which $\upsilon_+\approx \upsilon_-\approx
0$. The rapid shortening speeds were measured to $27~\mu$m$/$min and
$34~\mu$m$/$min for the plus and minus end, respectively, independently of the
surrounding tubulin concentration. We note that Janulevicius et
al.~\cite{Janul2006} use the values \eqref{eq:a+Walker}, \eqref{eq:b+Walker}
and $27~\mu$m$/$min (from Walker et al.~\cite{Walker1988}) for their model of
dynamic instability. The model consists of a growth phase with the velocity
given by \eqref{eq:v+}, a shrinkage phase with constant shrinkage speed
$27~\mu$m$/$min, and two affine relationships between the frequencies of rescue
and catastrophe as functions of the tubulin concentration, also with
coefficients from the measurements of Walker et al.~\cite{Walker1988}. These
frequencies are then used in a probability density function of an exponential
distribution of waiting times between different possible events.

To obtain the combined average effect of elongation and rapid shortening of the
plus and minus ends, Walker et al.~\cite{Walker1988} took into account the
average times for the different phases and obtained the affine relationships
shown in their Figure~9 for the plus and minus end, respectively. The sum of
these two functions gives an affine relationship for an average microtubule and
shows that the steady-state concentration is $10.9~\mu$M tubulin. Their
corresponding experimentally measured value was in the range 6.9--7.5~$\mu$M.
For the assembly of tubulin dimers in the growth cone of an axon, we are
interested in only the average growth of the plus-ends of microtubules. This
velocity as a function of tubulin concentration is given by a straight line in
\cite[Figure~9]{Walker1988}. We collect the following two points on the line:
$(c_1,\upsilon_1)=(14.43~\mu\text{M},73.42\text{~dimers}/\text{s})$ and
$(c_2,\upsilon_2)=(7.99~\mu\text{M},-113.04\text{~dimers}/\text{s})$. To
express the velocity in SI units, we assume that there are 1625 dimers per
$\mu$m of microtubule as Janulevicius et al.~\cite{Janul2006} do, i.e.\ a dimer
has the length $6.1538\cdot 10^{-10}$~m. This value agrees well with other
values in \cite{Walker1988} for which elongation speeds are expressed in both
m$/$s and dimers$/$s. In SI units, the two collected points are
$(c_1,\upsilon_1)=(0.01443~\text{mol}/\text{m}^3,4.518\cdot
10^{-8}\text{~m}/\text{s})$ and
$(c_2,\upsilon_2)=(0.00799~\text{mol}/\text{m}^3,-6.956\cdot
10^{-8}\text{~m}/\text{s})$ from which we obtain
\begin{align}\label{eq:rgparam}
\rg&=\frac{\upsilon_1-\upsilon_2}{c_1-c_2}=1.783\cdot 10^{-5}~\text{m}^4/(\text{mol\,s}),\\
\sg&=-(\upsilon_1-\rg c_1)=2.121\cdot 10^{-7}~\text{m}/\text{s},
\end{align}
and hence $\cinf=\sg/\rg=11.90\cdot 10^{-3}$~mol$/$m$^3$.

Janulevicius et al.~\cite{Janul2006} estimate from the data of Tanaka and
Sabry~\cite{Tanaka1995} that the volume of a growth cone $\Vc$ lies in the
range 1--200~$\mu$m$^3$. Approximating the growth cone by a sphere, this has
then a radius in the range 0.62--3.6~$\mu$m. We expect our parameter $\lc$ to
be of this order.

The remaining parameter in the dynamic model \eqref{eq:model} is the
polymerization reaction rate constant $\rgtilde$; see \eqref{eq:dldt}. This
constant appears only in the factor $(\rg\cc+\rgtilde\lc)$ in the equation for
the cone concentration in the model\eqref{eq:model}. To define a nominal value
we assume that $\rgtilde\lc$ is of the same order as $\rg\cc$. With
$\lc=4$~$\mu$m, $\cc=\cinf=11.9\cdot 10^{-3}$~mol$/$m$^3$ and $\rg$ given by
\eqref{eq:rgparam}, we get the nominal value $\rgtilde=0.053$~$s^{-1}$.

\section{All steady-state solutions and their properties}\label{sect:SS}

It is interesting to investigate the possible steady-state solutions of a
dynamic model. If the given soma concentration $\cs$ is constant and all
variables $c$, $l$, $\cc$ of the model equations \eqref{eq:model} are assumed
to be independent of time, then we denote the unknown constant axon length by
$\linf$, and note that $\cc=\cinf$ holds by the ODE for $l'(t)$. Then the model
equations \eqref{eq:model} yield the following linear boundary-value problem
for the unknown function $c=c(x)$ and the constant $\linf>0$:
\begin{equation}\label{eq:SSmodel}
\left\{
\begin{aligned}
&D\frac{\mathrm{d}^2c}{\mathrm{d} x^2}-a\dd{c}{x}-gc=0,&\quad&0<x<\linf,\\
&Dc'(\linf)=(a-g\lc)\cinf,&&\\
&c(\linf)=\cinf,&&\\
&c(0) = \cs.&&\\
\end{aligned}
\right.
\end{equation}

Note that the two parameters $\rg$ and $\rgtilde$ are not present in
\eqref{eq:SSmodel}, which means that they only influence dynamic solutions, not
steady-state solutions.

\subsection{The case $a>0$, $D>0$ and $g>0$}

The general solution of the ODE in \eqref{eq:SSmodel} is
\begin{align*}
&c(x)=k_+\rme^{\lambda_+x}+k_-\rme^{\lambda_-x},\\
&\text{where}\quad \lambda_{\pm}:=\frac{a\pm R}{2D}
\quad\text{and}\quad R:=\sqrt{{a^2}+4gD},
\end{align*}
and where $k_{\pm}$ are two constants. We note that
\begin{align*}
&\lambda_-<0<\lambda_+,\quad \lambda_+-\lambda_-=\frac{R}{D},\quad \lambda_++\lambda_-=\frac{a}{D},
\end{align*}
which we will use in the calculations below. The three boundary conditions in
\eqref{eq:SSmodel} give
\begin{equation}\label{eq:system1}
\left\{
\begin{alignedat}3
&Dk_+\lambda_+\rme^{\lambda_+\linf}&&+Dk_-\lambda_-\rme^{\lambda_-\linf}&&=(a-g\lc)\cinf\\
&k_+\rme^{\lambda_+\linf}&&+k_-\rme^{\lambda_-\linf}&&=\cinf \\
&k_+&&+k_-&&=\cs.
\end{alignedat}
\right.
\end{equation}
From the two upper equations of
\eqref{eq:system1} we can express $k_{\pm}$ in terms of $\linf$, for example,
by using Cramer's rule. The coefficient determinant is
\begin{align*}
&\Delta:=D(\lambda_+-\lambda_-)\rme^{\lambda_+\linf}\rme^{\lambda_-\linf}=R\rme^{a\linf/D}.
\end{align*}
Then we get
\begin{align*}
k_+&=\frac{1}{\Delta}
\begin{vmatrix}
(a-g\lc)\cinf &D\lambda_-\rme^{\lambda_-\linf}\\
\cinf&\rme^{\lambda_-\linf}
\end{vmatrix}
=\frac{\cinf\rme^{\lambda_-\linf}}{R\rme^{a\linf/D}}(a-g\lc-D\lambda_-)\\
&=\frac{\cinf\rme^{-\lambda_+\linf}}{R}\left(\frac{R+a}{2}-g\lc\right)
\end{align*}
and analogously
\begin{align*}
k_-&=\frac{1}{\Delta}
\begin{vmatrix}
D\lambda_+\rme^{\lambda_+\linf}&(a-g\lc)\cinf \\
\rme^{\lambda_+\linf}&\cinf
\end{vmatrix}
=\frac{\cinf\rme^{\lambda_+\linf}}{R\rme^{a\linf/D}}(D\lambda_+-a+g\lc)\\
&=\frac{\cinf\rme^{-\lambda_-\linf}}{R}\left(\frac{R-a}{2}+g\lc\right).
\end{align*}
Now the third equation of \eqref{eq:system1} can be written with only one
unknown, $\linf$:
\begin{align}\label{eq:linfeq}
&\frac{1}{R}\left[{\rme^{-\lambda_+\linf}}\left(\frac{R+a}{2}-g\lc\right)
+{\rme^{-\lambda_-\linf}}\left(\frac{R-a}{2}+g\lc\right)\right]=\frac{\cs}{\cinf}.
\end{align}
We can write this equation $f(\linf)={\cs}/\cinf$, where
\begin{align}
f(z)&:=\frac{1}{R}\left[{\rme^{-\lambda_+z}}\left(\frac{R+a}{2}-g\lc\right)
+{\rme^{-\lambda_-z}}\left(\frac{R-a}{2}+g\lc\right)\right],\label{eq:f}\\
f'(z)&=\frac{1}{R}\left[-\lambda_+{\rme^{-\lambda_+z}}\left(\frac{R+a}{2}-g\lc\right)
-\lambda_-{\rme^{-\lambda_-z}}\left(\frac{R-a}{2}+g\lc\right)\right].\label{eq:fprime}
\end{align}
Note that $f(0)=1$, which via \eqref{eq:linfeq} corresponds to $\cinf=\cs$. The
second terms within the squared brackets of both \eqref{eq:f} and
\eqref{eq:fprime} are positive (since $\lambda_-<0$ and $R>a$). This means that
both $f(z),f'(z)\rightarrow\infty$ (exponentially fast) as
$z\rightarrow\infty$. A graph of the function $f(z)$ given by \eqref{eq:f} is
shown in Figure~\ref{fig:f}.
\begin{figure}[tb]
  \small\centering\includegraphics[width=0.4\textwidth]{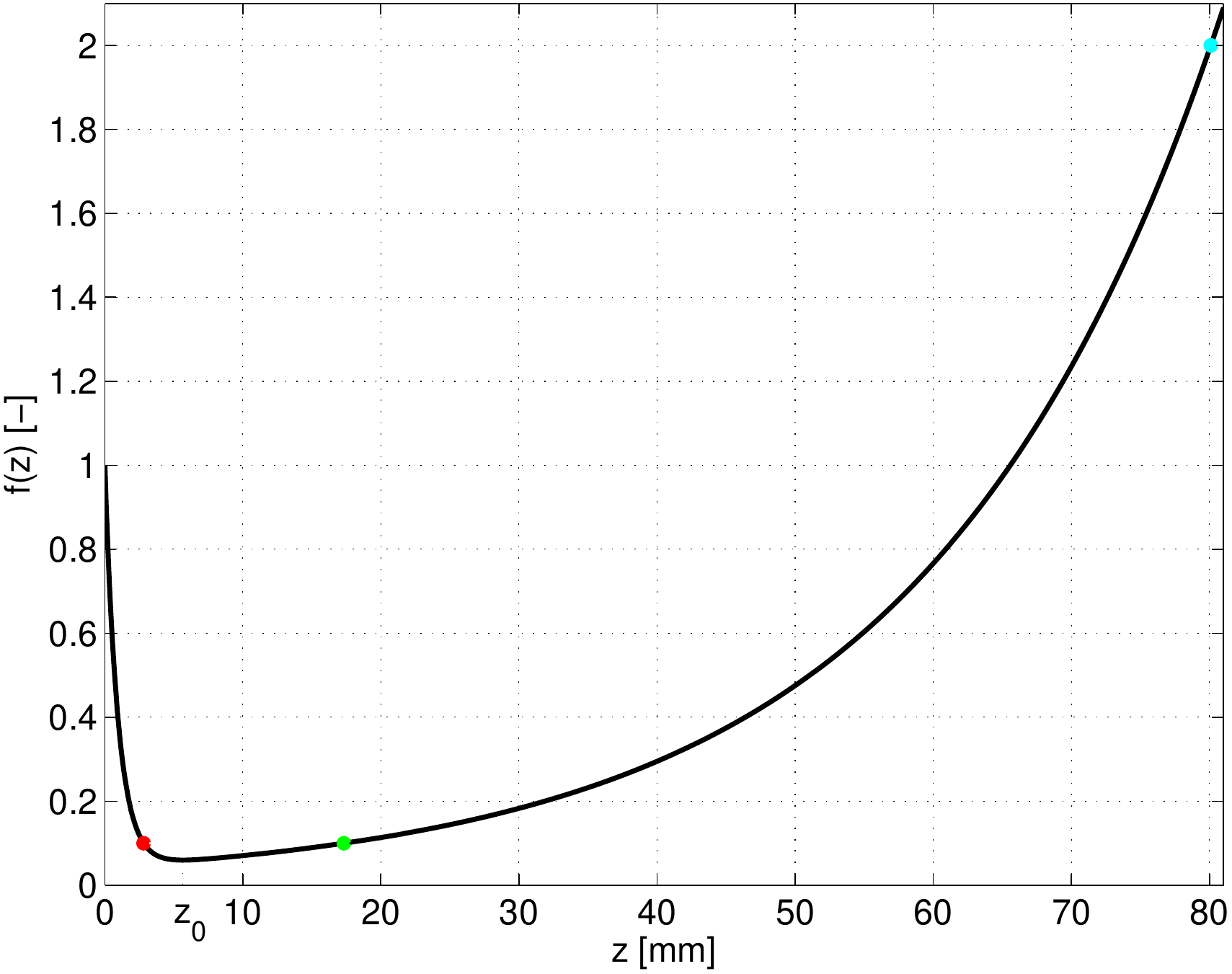}
  \caption{A graph of the function $f(z)$ given by \eqref{eq:f}
  for the nominal values of Table~\ref{table:1},
  which means that Case~I of Theorem~\ref{thm:case1} is valid.
  Note that the steady-state length $\linf$ is obtained from the equation
  $f(\linf)=\cs/\cc$. For a given value of the ratio $\cs/\cc$ on the $f(z)$-axis,
  one can easily read off the number of steady-state solutions and approximately the
  corresponding length(s) $\linf$ on the $z$-axis. For example, when $\cs/\cc=0.1$,
  there are two steady-state solutions corresponding to the red and green dots; see
  Figure~\ref{fig:iii}. As $\cs/\cc=2$, there is a unique steady-state solution (magenta dot);
  see Figure~\ref{fig:ivc}.
  }\label{fig:f}
\end{figure}%
We now investigate two main cases.

Assume first that $g\lc<(R+a)/2$ holds. Then and only then the first term
(within the squared brackets) of \eqref{eq:fprime} is negative, which is
equivalent to the fact that the equation $f'(z_0)=0$ has the unique solution
\begin{align}\label{eq:z_0}
z_0&=\frac{D}{R}\log\frac{(R+a)(R+a-2g\lc)}{(R-a)(R-a+2g\lc)}.
\end{align}
Since $f'(z)\rightarrow\infty$ as $z\rightarrow -\infty$, we can conclude that
$f(z)$ is a unimodal function with $z_0$ as its unique global minimum point;
see Figure~\ref{fig:f}. The nominal parameter values give
$$
g\lc=2\cdot 10^{-12}\text{~m$/$s}~<1\cdot
10^{-8}\text{~m$/$s~}=a.
$$
Then one can compute $R=1.10\cdot 10^{-8}\text{~m$/$s}$ and
\begin{align*}
z_0&=5.64\text{~mm}, \quad f(z_0)=0.060.
\end{align*}
The following equivalence is valid for the number $z_0$ given by
\eqref{eq:z_0}:
\begin{align*}
z_0>0&\quad\Longleftrightarrow\quad
\frac{(R+a)(R+a-2g\lc)}{(R-a)(R-a+2g\lc)}>1 \\
&\quad\Longleftrightarrow\quad g\lc<a.
\end{align*}
Hence, for $g\lc<a<(a+R)/2$ we can conclude that $f(z)$ is decreasing for
$0<z\le z_0$ and increasing for $z\ge z_0$. In particular, $f(z_0)<f(0)=1$
holds. Furthermore, since $g\lc<a<(a+R)/2$, it is clear that $f(z)>0$ for all
$z$, in particular, $f(z_0)>0$. Hence, we may have zero, one or two
steady-state solutions depending on the parameter values. In the subcase $a\le
g\lc<(R+a)/2\Leftrightarrow z_0\le 0$, we conclude that $f(z)$ is increasing
for $z\ge 0$, which implies that there exists a unique solution $\linf>0$ of
\eqref{eq:linfeq} if and only if $f(0)=1<\cs/\cinf$ holds.

Assume now that $g\lc\ge(R+a)/2$ holds. Then the first term of
\eqref{eq:fprime} is also positive, hence $f'(z)>0$ for all $z$, which implies
that \eqref{eq:linfeq} has a unique solution $\linf>0$ if and only if
$f(0)=1<\cs/\cinf$ holds. We conclude the cases in the following theorem.

\begin{theorem}\label{thm:case1}
Assume that $a>0$, $D>0$, $g>0$ hold. Then the following cases may occur:
\begin{itemize}
\item[I.] $g\lc<a$: There exists a $z_0>0$ given by \eqref{eq:z_0}. There
    are four subcases:
\begin{itemize}
\item[i.] $\cs/\cinf<f(z_0)<1$: No steady-state solution exists.
\item[ii.] $\cs/\cinf=f(z_0)$: There exists a unique steady-state
    solution of \eqref{eq:SSmodel} with $\linf=z_0$ and $c(x)$ given by
    the decreasing function
    \begin{multline}\label{eq:SSsolution}
c(x)=\frac{\cinf}{R}\left[
\left(\frac{R+a}{2}-g\lc\right)\rme^{\lambda_+(x-\linf)}\right.\\
\qquad\qquad
\left.+\left(\frac{R-a}{2}+g\lc\right)\rme^{\lambda_-(x-\linf)}
\right],\quad 0\le x\le\linf.
\end{multline}
\item[iii.] $f(z_0)<\cs/\cinf<1$: There exist two steady-state
    solutions. The axon lengths $\linfone$ and $\linftwo$ satisfy
    $0<\linfone<z_0<\linftwo$ and the corresponding concentration
    distributions $c_1(x)$ and $c_2(x)$ are given by
    \eqref{eq:SSsolution}. The function $c_1(x)$ is increasing for
    $0\le x\le\linfone$ and $c_2(x)$ is decreasing for $0\le x\le
    \linftwo-z_0$ and increasing for $\linftwo-z_0\le x\le\linf$.
\item[iv.] $\cs/\cinf\ge 1$: There exists a unique steady-state
    solution of \eqref{eq:SSmodel} with $\linf>z_0$, and $c(x)$ given
    by \eqref{eq:SSsolution} is decreasing for $0<x\le \linf-z_0$ and
    increasing for $\linf-z_0\le x\le\linf$.
\end{itemize}
\item[II.] $g\lc\ge a$: If and only if $\cs>\cinf$ holds, there exists a
    unique solution $\linf>0$ of \eqref{eq:linfeq} and a unique
    steady-state solution of \eqref{eq:SSmodel} given by
    \eqref{eq:SSsolution},
which is a decreasing function.
\end{itemize}
\end{theorem}

\begin{proof}
What remains to be proved are the monotonicity properties of $c(x)$ in the
different cases. Differentiation of \eqref{eq:SSsolution} gives
\begin{multline}\label{eq:cprime}
c'(x)=\frac{\cinf}{R}\left[
\left(\frac{R+a}{2}-g\lc\right)\lambda_+\rme^{\lambda_+(x-\linf)}\right.\\
\qquad\qquad
\left.+\left(\frac{R-a}{2}+g\lc\right)\lambda_-\rme^{\lambda_-(x-\linf)}
\right],\quad 0\le x\le\linf.
\end{multline}
We first note that, within the squared brackets, the second term is negative,
since $\lambda_-<0$. If $g\lc<(R+a)/2$, then the equation $c'(x_0)=0$ has the
unique solution
\begin{align}\label{eq:x_0}
x_0&=\linf-\frac{D}{R}\log\frac{(R+a)(R+a-2g\lc)}{(R-a)(R-a+2g\lc)}=\linf-z_0.
\end{align}
Furthermore, $c'(x)\lessgtr 0$ for $x\lessgtr x_0$. Note that this covers
Case~I, since $g\lc<a<(R+a)/2$. If $a\le g\lc<(R+a)/2$ holds, then $z_0\le 0$,
so that $x_0\ge\linf$. Hence $c(x)$ is decreasing for $0\le x\le\linf$. The
remaining case is $g\lc\ge (R+a)/2$. Then the first term of \eqref{eq:cprime}
(within the squared brackets) is negative, so that $c'(x)<0$ holds.
\end{proof}

With the non-zero values of the parameters given in Table~\ref{table:1}, the
different types of steady-state solutions are given by Theorem~\ref{thm:case1}.
We now demonstrate and comment on the different cases of
Theorem~\ref{thm:case1} by assigning $\cs$ different values, rather we let the
ratio $\cs/\cinf$ take different values. For the other parameters, we use the
nominal values given in \ref{table:1}. Note how the graph in Figure~\ref{fig:f}
of the auxiliary function $f(z)$ visualizes the different subcases of Case~I,
since given a value on the ratio $\cs/\cinf$ on the vertical axis, one can read
off the corresponding steady-state length(s) $\linf$ on the horizontal axis.

\paragraph{Case~I.i} If $\cs/\cinf<f(z_0)=0.060$, then
there exists no steady-state solution with $\linf>0$. As can be seen in
Figure~\ref{fig:f}, there exists no solution of the equation $f(z)=\cs/\cinf$
for values of $\cs/\cinf$ that lie below the minimum function value $f(z_0)$.

The biological interpretation is that when the soma concentration satisfies
$\cs<\cinf f(z_0)=7.12\cdot 10^{-4}$~mol$/$m$^3$, it is so small that no growth
can occur. If the soma concentration was larger earlier so that the axon has
reached a certain length and $cs$ drops to a value below $7.12\cdot
10^{-4}$~mol$/$m$^3$, then the dynamic equation for the axon length in
\eqref{eq:model} implies that $l'(t)<0$, i.e., the axon shrinks.

\paragraph{Case~I.ii}
This is a theoretical exceptional case with $\cs/\cinf=f(z_0)=0.060$, which
means that the only solution is the minimum point $z_0$ of $f$; hence the
steady state length is the small number $\linf=z_0=5.64\text{~mm}$. The
concentration distribution along the axon $c(x)$ is similar to the graph
$c_1(x)$ shown in Figure~\ref{fig:iii}.

The biological interpretation is that $\cs=\cinf f(z_0)=7.12\cdot
10^{-4}$~mol$/$m$^3$ is the smallest possible soma concentration that can
result in a stationary axon; however, its length is small.

\begin{figure}[tb]
  \small\centering\includegraphics[width=0.4\textwidth]{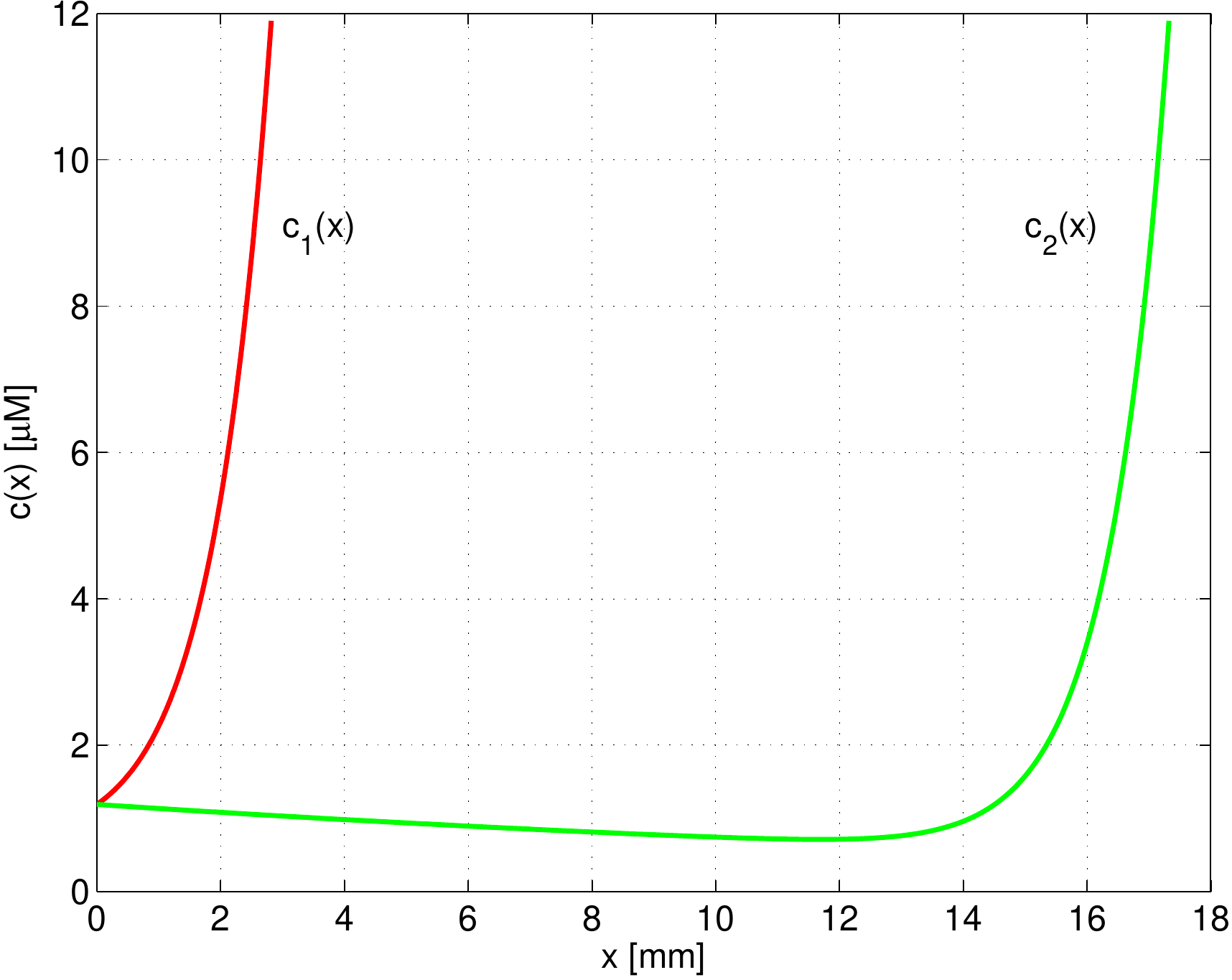}
  \caption{The two possible steady-state solutions in Case~I.iii with the soma concentration
  $\cs=1.19\cdot 10^{-3}$~mol$/$m$^3$, which is lower than the cone concentration
  $\cinf=11.9\cdot 10^{-3}$~mol$/$m$^3$.
  Compare with the coloured dots in Figure~\ref{fig:f}.}\label{fig:iii}
\end{figure}%

\paragraph{Case~I.iii}
In this case $\cs$ satisfies $f(z_0)=0.060<\cs/\cinf<1$. One can see in
Figure~\ref{fig:f} that the equation $f(z)=\cs/\cinf$ then has two solutions;
one is the small number $\linfone<z_0$ and the other $\linftwo>z_0$. As an
example, let $\cs/\cinf=0.1$, which means $\cs=1.19\cdot 10^{-3}$~mol$/$m$^3$.
Then Equation~\eqref{eq:linfeq} can be solved numerically to give
$\linfone=2.82$~mm and $\linftwo=17.3$~mm. The corresponding two concentration
distributions along the axon are shown in Figure~\ref{fig:iii}. In accordance
with Theorem~\ref{thm:case1}, $c_2(x)$ has a minimum point at
$x=\linftwo-z_0=11.7$~mm.

The biological interpretation is that when the soma concentration $\cs$ is
smaller than $\cinf$, then there are in fact two possible increasing
concentration distributions along the axon (hence two possible $\linf$), for
which the flux of positive active transport and negative diffusive transport is
balanced by the degradation along the axon.

\paragraph{Case~I.iv}
In this case $\cs>\cinf$, holds and there is a unique steady-state solution.
For example, $\cs=2\cinf=23.8\cdot 10^{-3}$~mol$/$m$^3$ implies that
$\linf=80.1$~mm and the concentration profile is shown in Figure~\ref{fig:ivc}.
\begin{figure}[tb]
  \small\centering\includegraphics[width=0.4\textwidth]{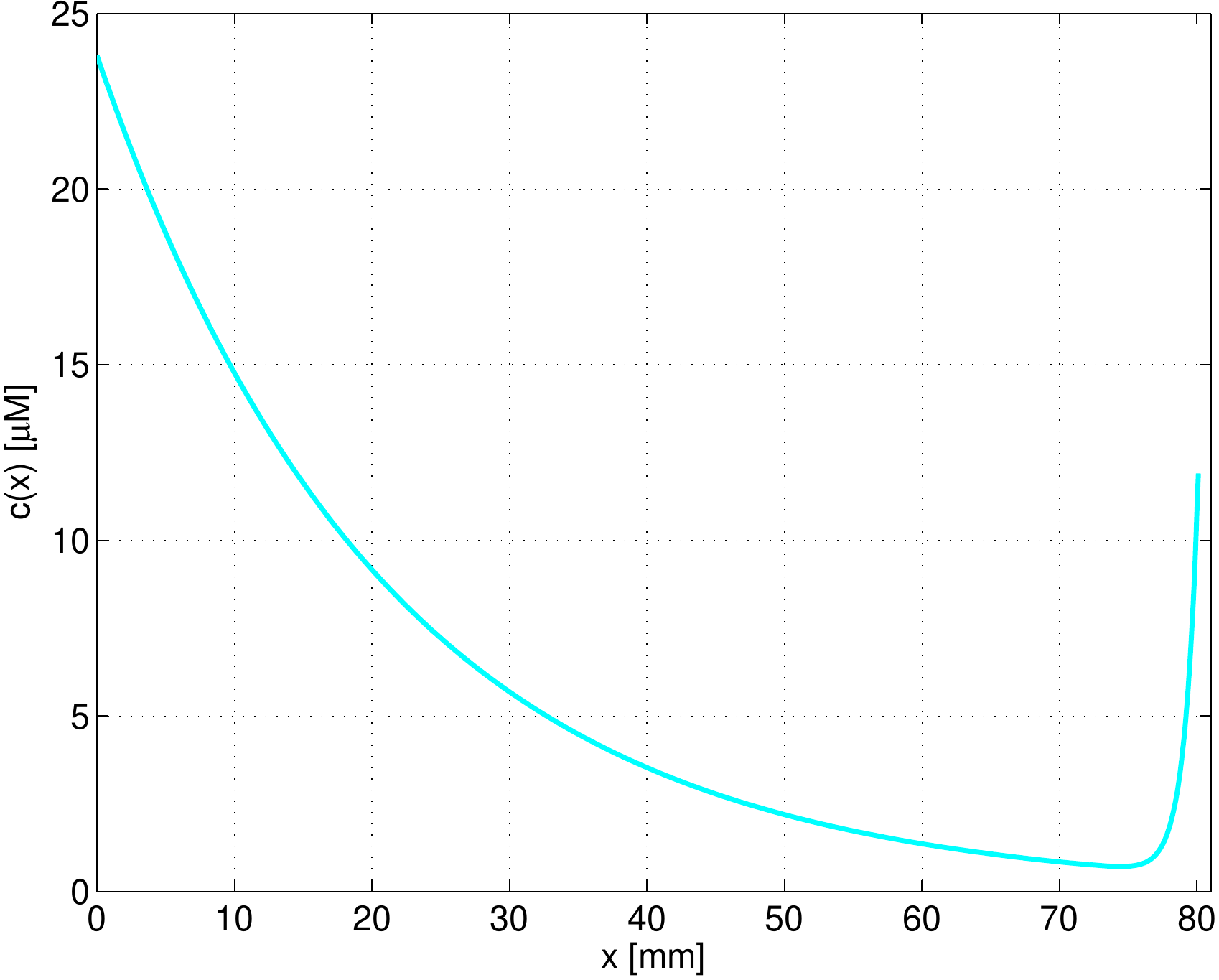}
  \caption{The Case~I.iv with the soma concentration
  $\cs=2\cinf=23.8\cdot 10^{-3}$~mol$/$m$^3$.
  Compare with the cyan dot in Figure~\ref{fig:f}.
  }\label{fig:ivc}
\end{figure}%

The biological interpretation is that this is the normal case with a larger
soma concentration than growth cone concentration. Note the non-monotone
concentration $c(x)$ along the axon and that the concentration is lower than
the growth cone concentration $\cinf$ along most of the axon length (the
interval between $x\approx 30$~mm and $x=\linf=80.1$~mm). The curve $c_2(x)$ in
Figure~\ref{fig:iii} also has the same principle non-monotone form. The
explanation for this non-monotone concentration profile $c(x)$ is that this
very form is precisely such that all the following effects are precisely
balanced: active transport, diffusion, degradation along the axon and in the
growth cone, assembly and disassembly of tubulin dimers. All these effects are
combined and it is therefore difficult to look at one or two separately in
order to get an intuitive feeling for the form of $c(x)$. However, we discuss
this further in Section~\ref{sect:discussion}.

\paragraph{Case~II}
In this case, the active transport coefficient satisfies $a<g\lc=2\cdot
10^{-12}\text{~m$/$s}$. Such small values of $a$ yield, from a biological
points of view, indistinguishable results from the case with $a=0$, which is is
dealt with in the next subsection.

\subsection{The case $a=0$, $D>0$ and $g>0$}

This case with no advective transport is a special case of the one above with
$R=2\sqrt{gD}$ and $\lambda_{\pm}=\pm \sqrt{g/D}$. It is easy to conclude that
$z_0<0$. We are then in Case~II of Theorem~\ref{thm:case1}.

\begin{theorem}\label{thm:case2}
Assume that $a=0$, $D>0$ and $g>0$ hold. If and only if $\cs>\cinf$ holds there
exists a unique solution $\linf>0$ of \eqref{eq:linfeq} and a unique
steady-state solution of \eqref{eq:SSmodel} given by
\begin{multline}\label{eq:ccase2}
c(x)=\frac{\cinf}{2}\left[
\left(1-\lc\sqrt{\frac{g}{D}}\right)\rme^{(x-\linf)\sqrt{\frac{g}{D}}}\right.\\
\qquad\qquad
\left.+\left(1+\lc\sqrt{\frac{g}{D}}\right)\rme^{-(x-\linf)\sqrt{\frac{g}{D}}}
\right],\quad 0\le x\le\linf,
\end{multline}
which is a decreasing function.
\end{theorem}

In this case the only transport of tubulin from the soma to the growth cone is
diffusion, wherefore a decreasing concentration distribution $c(x)$ is the only
possibility (diffusion occurs from higher to lower concentrations). As an
example, we choose $\cs=4\cinf=47.6\cdot 10^{-3}$~mol$/$m$^3$, which yields the
plots of Figure~\ref{fig:IIc}. The upper plot shows the auxiliary function
$f(z)$, which is now increasing. As is shown by the magenta dot; despite a soma
concentration $\cs$ four times higher than $\cinf$, the steady-state length is
only $\linf=9.22$~mm. The corresponding concentration profile $c(x)$ given by
\eqref{eq:ccase2} is shown in the lower plot of Figure~\ref{fig:IIc}.
\begin{figure}[tb]
  \small
  \centering\includegraphics[width=0.4\textwidth]{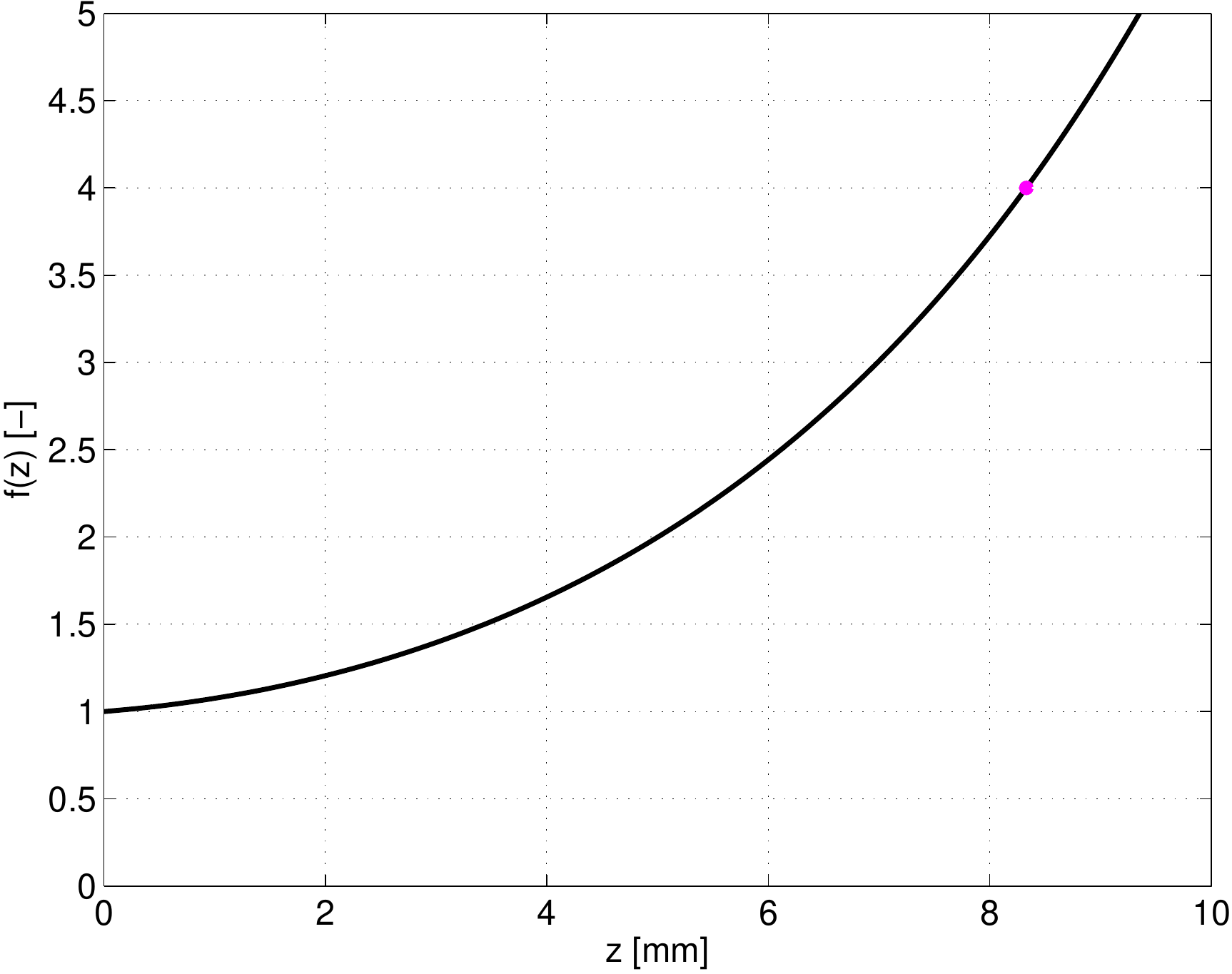}
  \centering\includegraphics[width=0.4\textwidth]{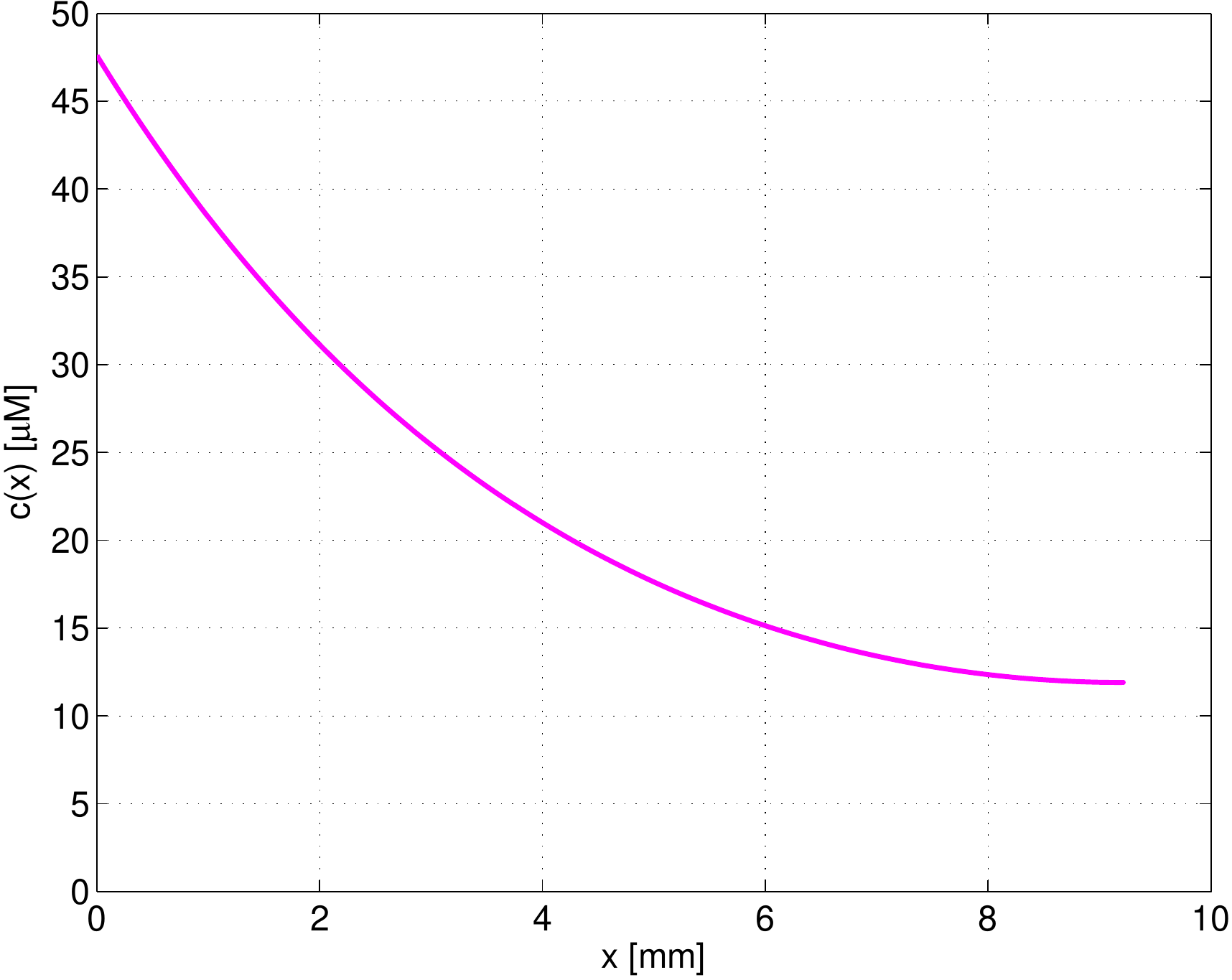}
  \caption{The case when $a=0$ (Theorem~\ref{thm:case2} or
  Case II of Theorem~\ref{thm:case1})
  with the ratio
  $\cs/\cinf=4$. Since the only transportation of tubulin occurs
  with diffusion, the steady-state length is much shorter than when active
  transport also is present ($a>0$).}\label{fig:IIc}
\end{figure}%

\subsection{The case $a>0$, $D>0$ and $g=0$}

In this case we have $R=a$, $\lambda_+=a/D$ and $\lambda_-=0$. Then
\eqref{eq:linfeq} is reduced to
$$
a\mathrm{e}^{-a\linf/D}=\frac{R\cs}{\cinf}\quad\Longleftrightarrow\quad
\linf=\frac{D}{a}\log\frac{\cinf}{\cs}.
$$
\begin{theorem}\label{thm:case3}
Assume that $a>0$, $D>0$ and $g=0$. If and only if $\cs<\cinf$ holds there
exists a unique solution $\linf>0$ of \eqref{eq:linfeq} and a unique
steady-state solution of \eqref{eq:SSmodel} given by
\begin{align*}
c(x)={\cinf}\rme^{(x-\linf)a/D}=\cs\rme^{xa/D},\quad 0<x\le\linf=\frac{D}{a}\log\frac{\cinf}{\cs}.
\end{align*}
\end{theorem}
Since there is no degradation of tubulin, the total flux is zero at every point
of the $x$-axis. This means that the active transport in the direction of
increasing $x$-values is precisely balanced by diffusion in the opposite
direction, which is possible if and only if $c(x)$ is increasing.

\subsection{The case $D>0$ and $a=g=0$}

In this case, the solution of the ODE in \eqref{eq:SSmodel} is the affine
function $c(x)=k_1x+k_2$, where $k_1$ and $k_2$ are constants to be determined
by the boundary conditions of \eqref{eq:SSmodel}. One finds that for any value
of $\linf>0$ the only possible solution is
\begin{align}\label{eq:constsol}
&c(x)=\cs=\cinf,\quad 0<x\le\linf.
\end{align}
Hence, there is neither any active ($a=0$) nor any diffusive transport, since
the concentration is the same along the axon; the diffusive flux is
$-Dc_x=-D\cdot 0=0$. Since there is no degradation, this case has probably no
biological interest.

\subsection{The cases with $D=0$}\label{sect:D=0}

If there is no diffusion present, but $a>0$ and $g>0$, then the first boundary
condition in \eqref{eq:SSmodel} requires that $a=g\lc$ holds, which one cannot
expect to be fulfilled. However, if $a=g\lc$ holds, then the steady-state
solution satisfying the boundary conditions is
\begin{align*}
c(x)={\cinf}\rme^{(\linf-x)/\lc}=\cs\rme^{-x/\lc},\quad 0<x\le\linf=\lc\log\frac{\cs}{\cinf}.
\end{align*}
Hence, $\linf>0$ if and only if $\cs>\cinf$, and then the concentration is
decreasing along the axon.

In fact, when diffusion is not present, the extra assumption of a continuous
concentration distribution is unnatural. When active transport is the only
movement of a substance, then all waves, including discontinuities, in the
concentration profile are transported with the velocity $a$. In fact, the PDE
of \eqref{eq:model} is hyperbolic and solutions of such may contain
discontinuities travelling with the speed $a$ during dynamic situations. When
there is no diffusion that can smooth out sharp gradients, one cannot exclude
the case $c^-\neq\cc$. Removing the continuity assumption,
Equation~\eqref{eq:dccdt1} should be used for the dynamics of the cone
concentration $\cc$ instead of the second equation of \eqref{eq:model} where
the assumption $c^-=\cc$ have been used. Then the boundary condition at
$x=l(t)$ in \eqref{eq:SSmodel} should be replaced by $Dc_x^-=ac^--g\lc\cc$.
Hence, when $D=0$ we have $ac^-=g\lc\cc$. This is a much more flexible
condition that is fulfilled for all biologically possible parameter values. The
steady-state problem is then reduced to
\begin{equation}\label{eq:SSdisc}
\left\{
\begin{aligned}
&-a\dd{c}{x}-gc=0,&\quad&0<x<\linf,\\
&ac^-=g\lc\cinf,&&\\
&c(0) = \cs.&&\\
\end{aligned}
\right.
\end{equation}
Assume that $a>0$ and $g>0$. We first note that since any discontinuity in the
interval $(0,l(t))$ has the positive speed $a$, there exists no stationary
discontinuity. The only possible discontinuity in steady state is at $x=\linf$.
Straightforward calculations give the following theorem.
\begin{theorem}\label{thm:SSdisc}
Assume that $D=0$, $a>0$ and $g>0$ and consider the axonal-growth problem
without the assumption that the concentration $c(x,t)$ is a continuous
function, i.e., problem \eqref{eq:model} with the ODE for $\cc$ replaced by
\eqref{eq:dccdt1}. If and only if $\cs>\cinf g\lc/a$ holds there exists a
unique steady-state solution of \eqref{eq:SSdisc} given by
\begin{equation}\label{eq:discsol}
c(x)=\cs\rme^{-gx/a},\quad\text{and}\quad\linf=\frac{a}{g}\log\frac{a\cs}{g\lc\cinf}.
\end{equation}
\end{theorem}
Note that $c(\linf)=\cs\rme^{-g\linf/a}=g\lc\cinf/a\neq\cinf$, unless $a=g\lc$
holds, which is the special case above when the solution is continuous. The
cases when also either $a=0$ or $g=0$ are trivial and biologically
uninteresting.

\section{The steady-state solutions' dependence on each parameter}\label{sect:plotsSSpar}

Given the nominal parameter values of a steady-state solution (see
Table~\ref{table:1}), we shall now investigate the sensitivity of the
steady-state solutions with respect to each parameter.

\begin{figure}[tb]
  \small
  \centering\includegraphics[width=0.4\textwidth]{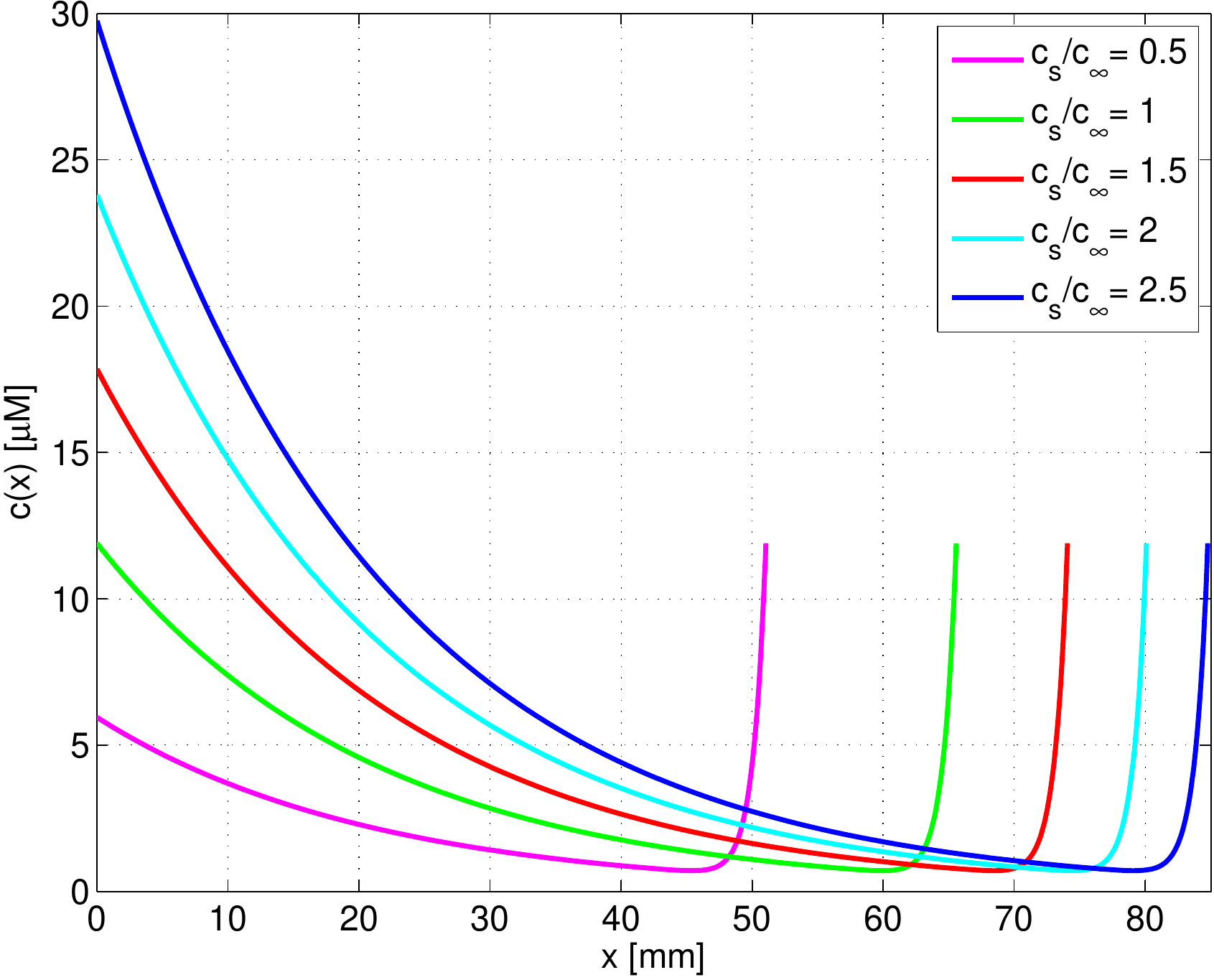}
  \caption{Concentration profiles for varying $\cs$.
  Recall that $\cinf=11.9~\mu$M.
  The magenta graph ($\cs/\cinf=0.5$)
  is $c_2(x)$ of Case~iii and the other belong to
  Case~iv.}\label{fig:varyc0}
\end{figure}%

In Figure~\ref{fig:varyc0}, the concentration profiles along the axon are shown
when $\cs$ is varied and the other parameters are the nominal ones. It is
interesting to note that for the nominal values, the steady-state length
$\linf=65.6$~mm when the soma concentration $\cs$ is equal to the steady-state
cone concentration $\cinf=11.9~\mu$M.

\begin{figure}[tb]
  \small
  \centering\includegraphics[width=0.4\textwidth]{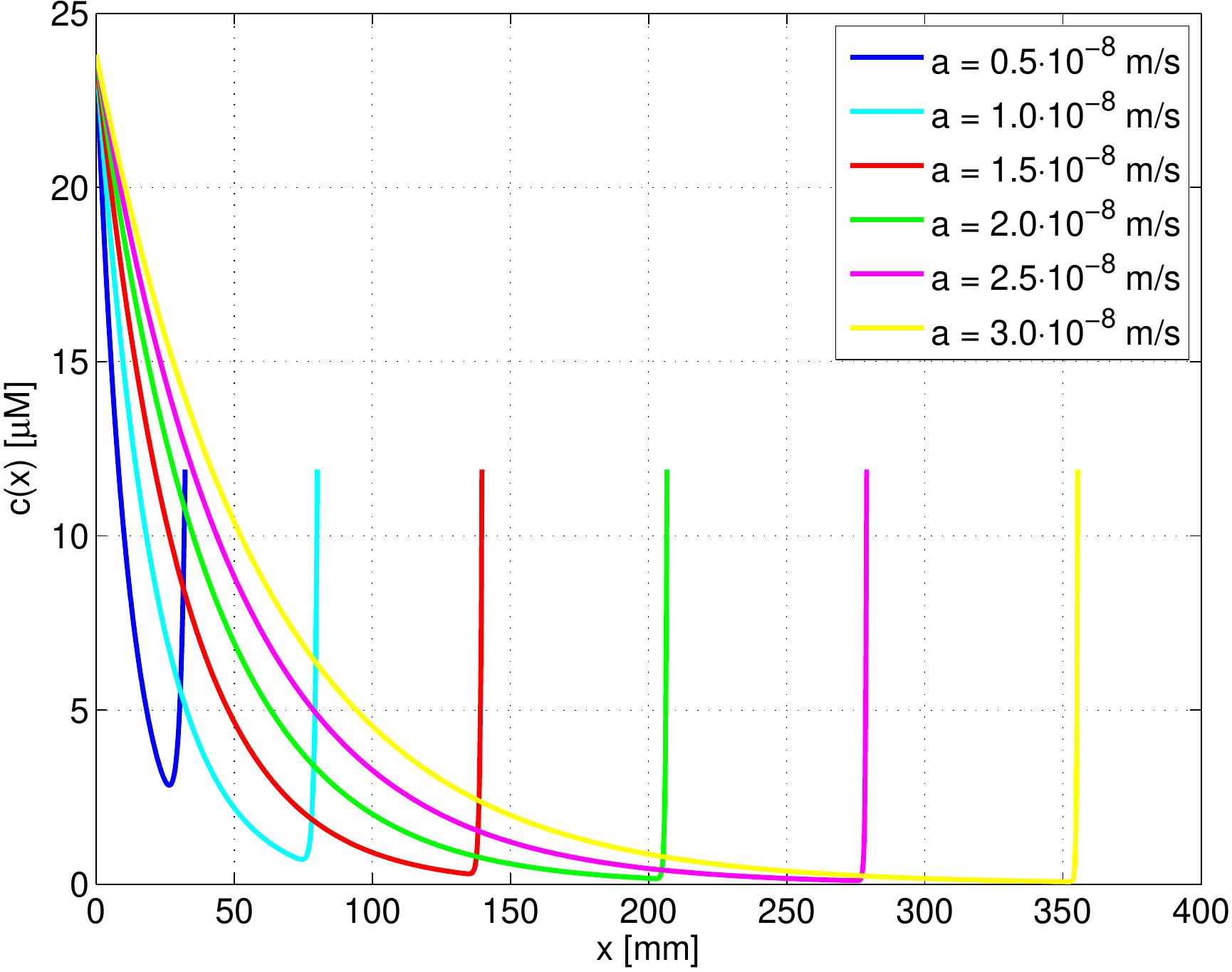}
  \includegraphics[width=0.4\textwidth]{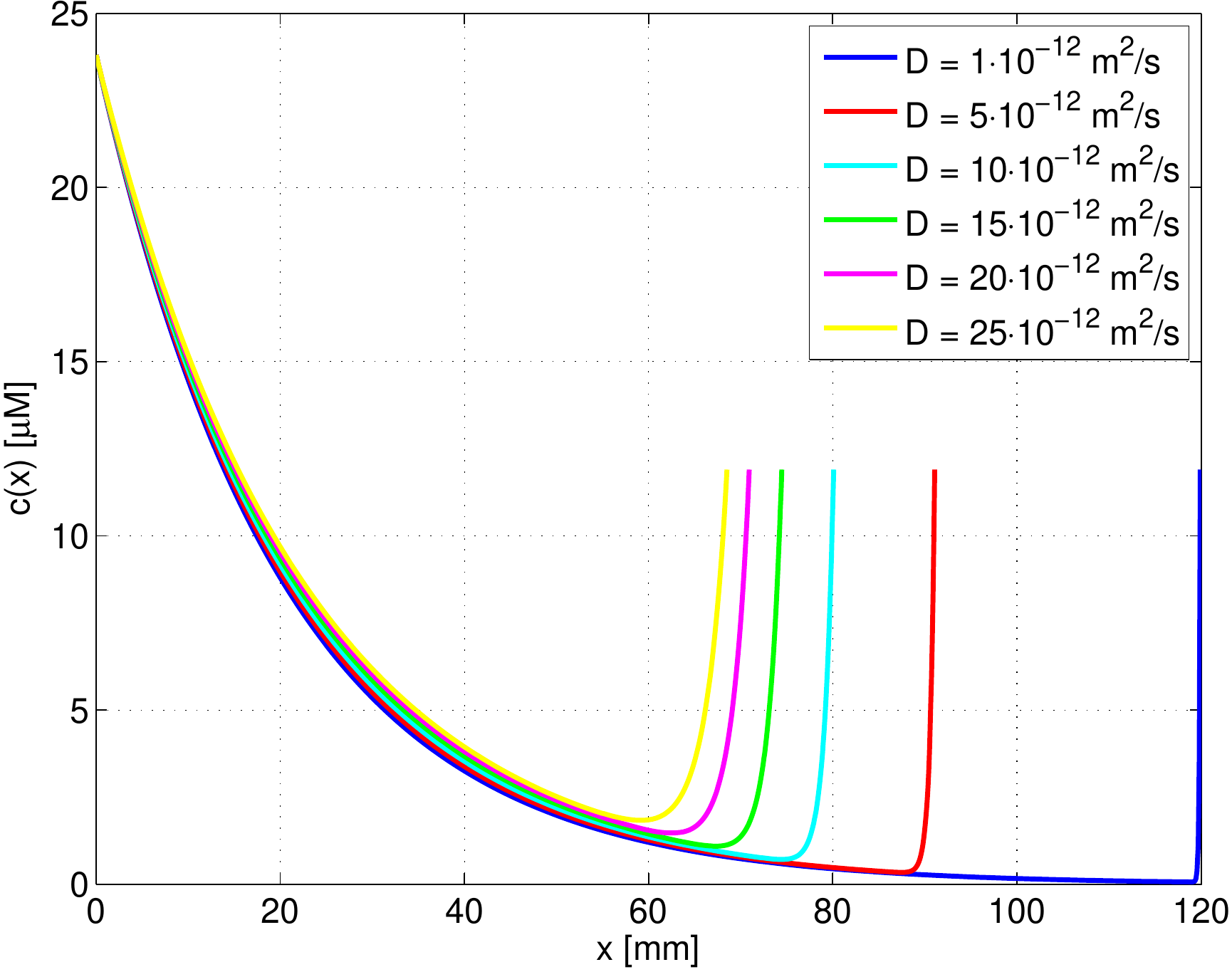}
  \caption{Concentration profiles in Case~I.iv for varying
  $a$ and $D$, respectively.}\label{fig:varya}
\end{figure}%

Now we keep the ratio $\cs/\cinf=2$, i.e., $\cs=2\cinf=23.8\cdot
10^{-3}$~mol$/$m$^3$ and vary the other steady-state parameters. This
corresponds to Case~I.iv in Theorem~\ref{thm:case1}.

\begin{figure}[tb]
  \small
  \centering
  \includegraphics[width=0.4\textwidth]{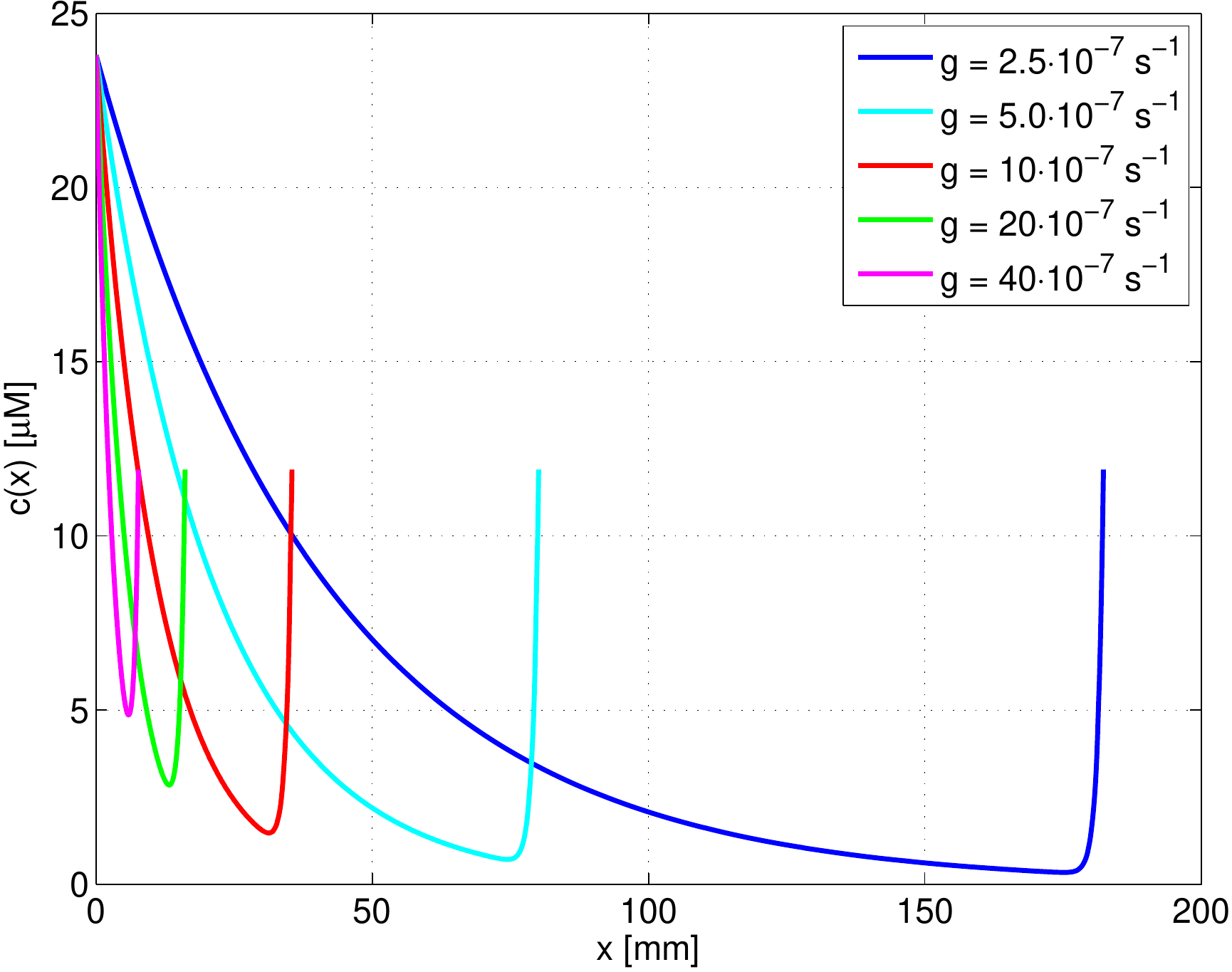}
  \includegraphics[width=0.4\textwidth]{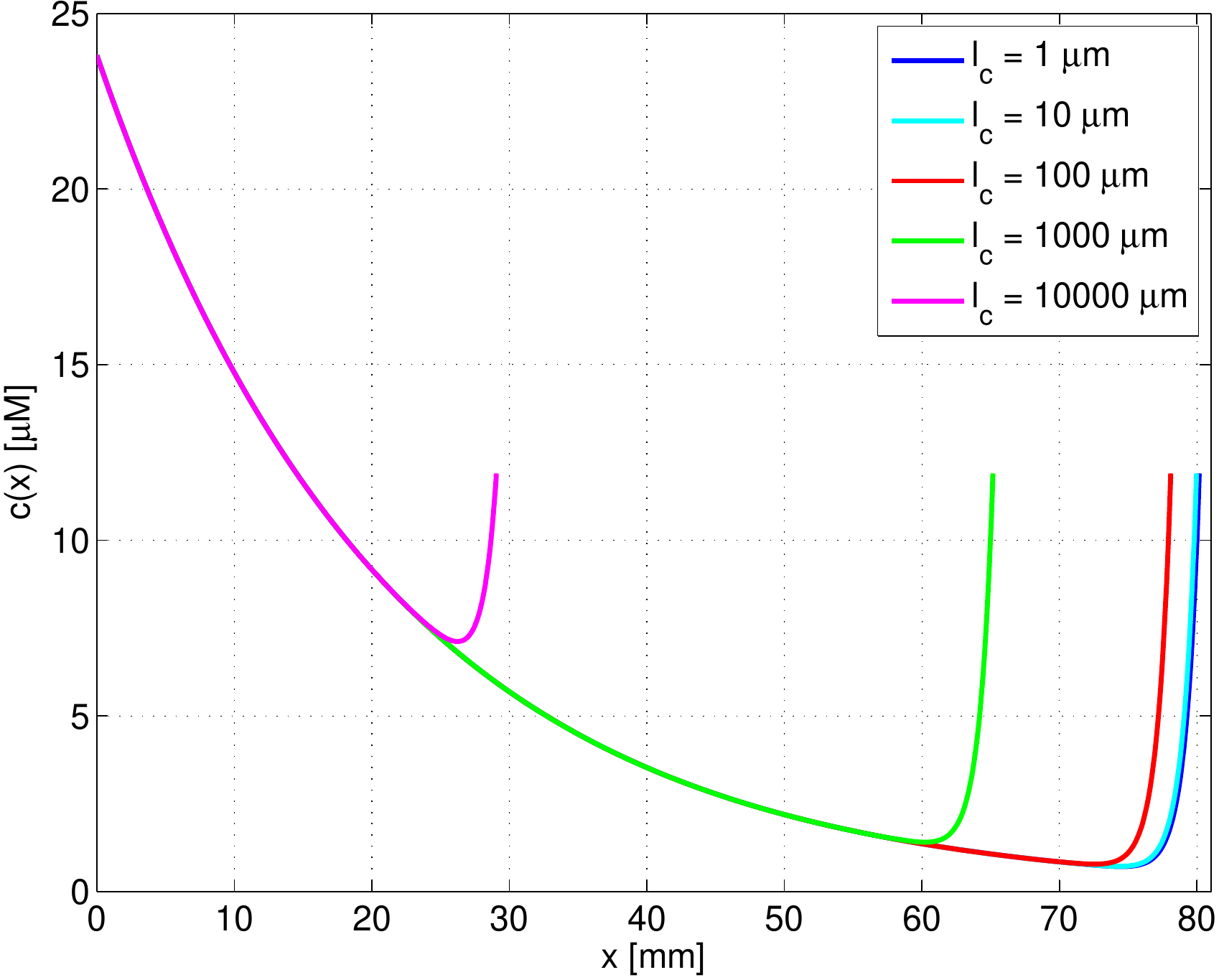}
  \caption{Concentration profiles in Case~I.iv for varying
  $g$ and $\lc$, respectively.}\label{fig:varyg}
\end{figure}%

From Figure~\ref{fig:varya} we can conclude that the steady-state length
$\linf$ increases with the active transport $a$ and both $\linf$ and the
concentration profile $c(x)$ are sensitive to small variations in $a$. On the
other hand, increasing the diffusion $D$ implies a decrease in $\linf$, but
hardly changes $c(x)$, except near the growth cone. Decreasing $D$ means a
substantial increase in $\linf$. Note that, according to
Section~\ref{sect:D=0}, in the limit $D=0$ there exists generally no continuous
steady-state solution. The physically relevant solution is given by
\eqref{eq:SSdisc}, which is a decreasing function all the way to the
steady-state length $\linf=184$~mm (for $\cs/\cinf=2$).

Given the nominal values of Table~\ref{table:1}, Figure~\ref{fig:varyg} shows
that increasing the degradation rate $g$ implies a substantial decrease in
$\linf$. However; the length $\linf$ is relatively insensitive to the size of
the growth cone expressed by the parameter $\lc$, unless this gets very big.

\section{On the stability of the steady states}\label{sect:stability}

It is of interest to know whether a steady-state solution of a mathematical
model is biologically and physically relevant, i.e., whether it can exist in
reality. A steady-state solution of a model is stable if one uses a disturbance
of the steady state as initial data and the dynamic solution converges back to
the steady state. If the solution moves further away, the steady state is
unstable.

Referring to the different cases of parameter values treated in the theorems of
Section~\ref{sect:SS}, we expect that when there exists a unique steady state,
it is stable. However, in Case~I.iii of Theorem~\ref{thm:case1}, in which
$f(z_0)=0.06<\cs/\cinf<1$, there exist two steady states and the question is
what happens for large times. Because of the presence of diffusion, which has a
damping effect on any oscillation, it is reasonable to assume that the dynamic
solution converges to a steady state as time increases.

Since it is difficult to make stability analyses mathematically, we make the
investigations here numerically. McLean and Graham~\cite{McLean2004} make a
spatial transformation so that the moving interval $(0,l(t))$ for the PDE is
transformed to the fixed interval $(0,1)$. A numerical implementation for their
problem is presented by Graham et al.~\cite{Graham2006}. The transformation
means that the same number of spatial computational cells is used along the
axon irrespective of its length. We will use the same spatial transformation:
\begin{align*}
&y:=\frac{x}{l(t)},\qquad \pp{y}{x}=\frac{1}{l(t)},\qquad
\pp{y}{t}=-\frac{xl'(t)}{l(t)^2}=-\frac{yl'(t)}{l(t)}.
\end{align*}
With $\cb(y,t):=c(yl(t),t)$, we get
\begin{align*}
&\pp{c}{t}=\pp{\cb}{t}-\frac{yl'(t)}{l(t)}\pp{\cb}{y},\qquad
\pp{c}{x}=\frac{1}{l(t)}\pp{\cb}{y},\qquad
\frac{\partial^2c}{\partial x^2}=\frac{1}{l(t)^2}\frac{\partial^2\cb}{\partial y^2}.
\end{align*}
Then the dynamic model \eqref{eq:model} is transformed to
\begin{equation}\label{eq:modely}
\left\{
\begin{aligned}
&\pp{\cb}{t}+\alpha(y,\cc,l)\pp{\cb}{y}
-\frac{D}{l^2}\frac{\partial^2\cb}{\partial y^2} = -g\cb,&\quad&0<y<1,\ t>0,\\
&\begin{split}
&\dd{\cc}{t}=\frac{(a-g\lc)}{\lc}\cc-\frac{D}{\lc l}\cb_y^-\\
&\qquad - \frac{(\rg\cc+\rgtilde\lc)}{\lc}(\cc-\cinf),
\end{split}&& t>0,\\
&\dd{l}{t}=\rg(\cc-\cinf),&&t>0,\\
&\cb(0,t) = \cs(t),&&t\ge 0,\\
&\cb(1,t)= \cc(t),&&t>0,\\
&\cb(y,0) = c_0(yl_0),&&0\le y\le 1,\\
&\cc(0)= c_0(l_0),&&\\
&l(0)=l_0,&&
\end{aligned}
\right.
\end{equation}
where
\begin{align*}
\alpha\big(y,\cc(t),l(t)\big):=\frac{a-yl'(t)}{l(t)}=\frac{a-y\rg\big(\cc(t)-\cinf\big)}{l(t)}.
\end{align*}
We will find approximate solutions to the model~\eqref{eq:modely} using the
method of lines by performing a spatial discretization of the PDE. The
$y$-interval $[0,1]$ is divided into $M$ subintervals all of the size $\Delta
y:=1/M$. Set $y_j:=j\Delta y$ and let $\Cb_j \approx \cb(y_j,t)$ for
$j=0,\ldots,M$, where $\Cb_M(t)=\cc(t)$ by the continuity boundary condition.
Spatial second-order difference approximations are used:
\begin{equation*}
\pp{\cb}{y}(y_j,\cdot) \approx \frac{\Cb_{j+1} - \Cb_{j-1}}{2\Delta y},
\quad \frac{\partial^2\cb}{\partial y^2}(y_j,\cdot) \approx \frac{\Cb_{j+1} - 2\Cb_{j} + \Cb_{j-1}}{(\Delta y)^2}.
\end{equation*}
In the ODE for the cone concentration, we use the one-sided second order
approximation
\begin{equation*}
\cb_y^- \approx \frac{\Cb_M - 4\Cb_{M-1} + 3\Cb_{M-2}}{2\Delta y} = \frac{\cc - 4\Cb_{M-1} + 3\Cb_{M-2}}{2\Delta y}.
\label{eq:FD_end_point}
\end{equation*}
Denote the time step by $\Delta t$ and set $t^n:=n\Delta t$, $n=0,1,\ldots$. At
time $t=t^n$, the concentration within the axon is approximated by the
numerically computed values $\Cb_j^n\approx \cb(y_j,t^n)$, $j=1,\ldots,M-1$.
The approximate growth-cone concentration is denoted by
$\Cc^n\approx\ccb(t^n)$, the axon length by $L^n\approx \lb(t^n)$. The explicit
Euler method means that each time derivative is approximated by formulas like
\begin{align*}
&\pp{\cb}{t}(y_j,t^n)\approx\frac{\Cb_j^{n+1} - \Cb_j^n}{\Delta t}.
\end{align*}
As initial values, we set
\begin{align*}
&\Cb_j^0 = c_0(y_jl_0), \quad j=1,\ldots,M, \quad \Cc^{0}=c_0(l_0) \quad\text{and}\quad L^0=l_0.
\end{align*}
Substituting the approximations of the derivatives into \eqref{eq:modely} we
get an explicit time marching numerical method, i.e., only old values of the
unknowns are used. For example, the update formulas for the concentration
values along the axon are, for $j=1,\ldots,M-1$:
\begin{align}
\begin{split}
\Cb_j^{n+1}&=
\Cb_j^{n}+\Delta t\left(-\alpha(y_i,\Cc^n,L^n)\frac{\Cb^n_{j+1} - \Cb^n_{j-1}}{2\Delta y}\right.\\
&\qquad\qquad+\left.\frac{D}{(L^n)^2}\frac{\Cb_{j+1}^n - 2\Cb_{j}^n + \Cb_{j-1}^n}{(\Delta y)^2}-g\Cb_j^n\right)
\end{split}\notag\\
\begin{split}
&=\left(1-\Delta t\left(\frac{2D}{(L^n\Delta y)^2}-g\right)\right)\Cb_j^{n}\\
&\qquad+\Delta t\left(-\frac{\alpha(y_i,\Cc^n,L^n)}{2\Delta y}+\frac{D}{(L^n\Delta y)^2}\right)\Cb^n_{j+1}\\
&\qquad\qquad+\Delta t\left(\frac{\alpha(y_i,\Cc^n,L^n)}{2\Delta y}+\frac{D}{(L^n\Delta y)^2}\right)\Cb^n_{j-1}.
\end{split}\label{eq:PDEupdate}
\end{align}

Given $\Delta y=1/M$, the time step $\Delta t$ has to be chosen sufficiently
small to avoid instabilities in the updates \eqref{eq:PDEupdate} for the
advection-diffusion PDE of \eqref{eq:modely}. For example, with constant
Dirichlet boundary conditions, there exists a standard so called CFL condition
to ensure stability. The complication here is the coupling to the two ODEs via
the boundary at $x=l(t)$. Therefore, we make some assumptions on the numerical
updates. To be more precise, let $N$ be the number time steps, $T:=N\Delta t$
the total simulation time and set
\begin{align*}
\cmax&:=\max\big(\cinf,\max\limits_{0\le t\le T}\cs(t)\big).
\end{align*}
Assume that the numerical values satisfy
\begin{align}\label{eq:Lbounded}
&L^n \ge\lmin>0\quad\text{and}\quad0\le\Cc^n\le\cmax\quad\text{for $n=0,1,\ldots,N$},
\end{align}
where $\lmin$ is a constant. The scheme \eqref{eq:PDEupdate} is monotone (or
positive) if the coefficients for $\Cb^n_{j-1}$, $\Cb^n_{j}$ and $\Cb^n_{j+1}$
are non-negative. Then no unphysical numerical oscillations appear within the
axon. The requirements on the discretization parameters are the cell P\'{e}clet
condition for $\Delta y$ and a CFL condition for $\Delta t$:
\begin{align}
\Delta y&\le\frac{2D}{(a+\rg2\cmax)\lmin},\label{eq:stability}\\
\Delta t&\le\left(g+\frac{2D}{(\lmin\Delta y)^2}\right)^{-1}.\label{eq:CFL}
\end{align}

We first demonstrate Case~I.iv of Theorem~\ref{thm:case1}, which is the case
when $\cs/\cinf>1$ and there exists a unique steady-state solution. To be able
to compare with Figure~\ref{fig:ivc}, we choose the constant ratio
$\cs/\cinf=2$ and the initial axon length $l(0)=2$~mm. With the spatial
discretization $M=1/\Delta y=1000$, the stability criterion
\eqref{eq:stability} is satisfied (with $\lmin=2$~mm) and the CFL condition
\eqref{eq:CFL} gives $\Delta t\le 5.0$~s. The axon length as function of time
is shown in Figure~\ref{fig:stabkvot2}.
\begin{figure}[tbh]
  \centering\includegraphics[width=0.4\textwidth]{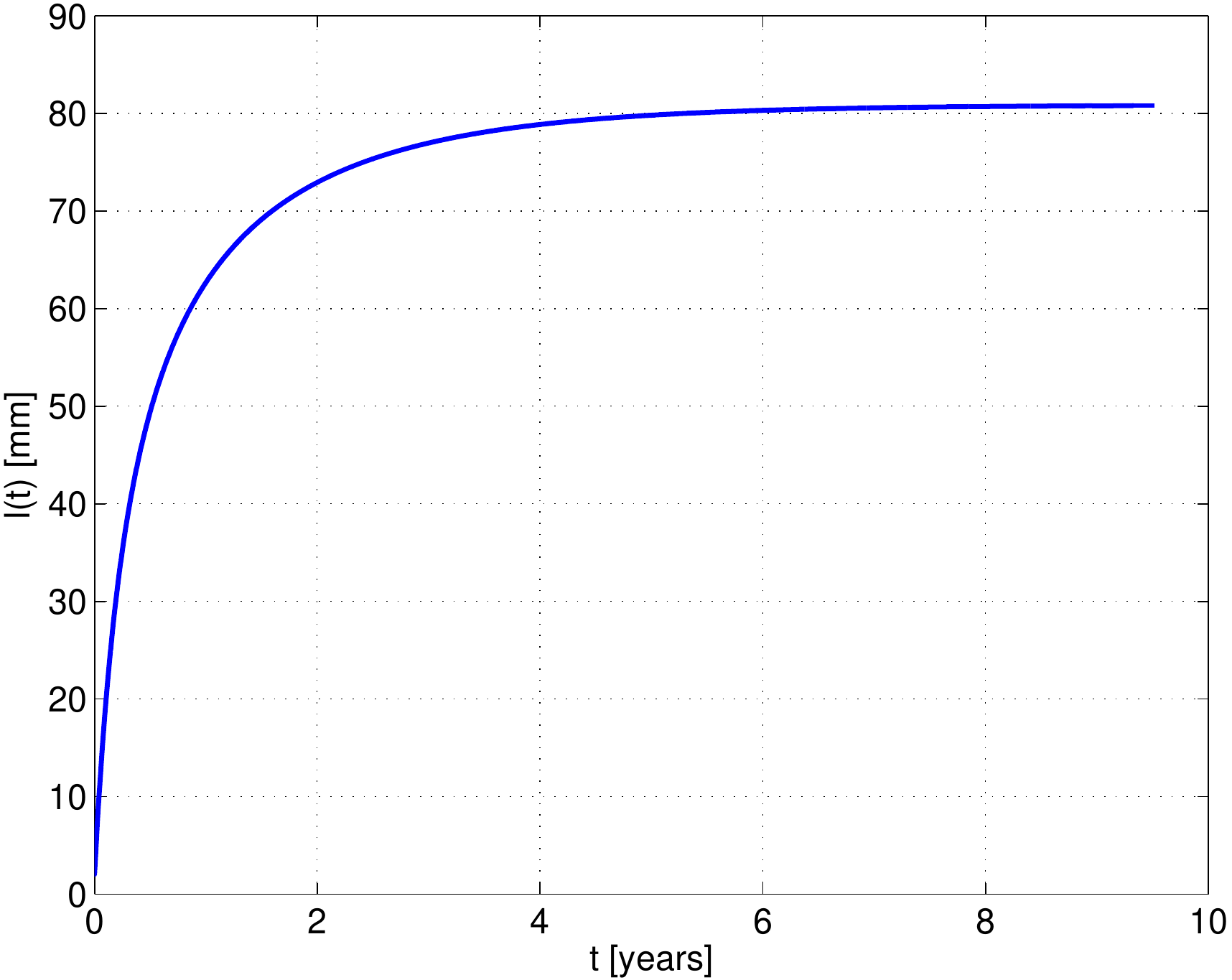}
  \caption{A dynamic simulation of the outgrowth of an axon when $\cs/\cinf=2$.
  According to the steady-state results in Section~\ref{sect:SS},
  the final length is $\linf=80.1$~mm and the final concentration
  distribution is the same as the theoretically obtained one
  in Figure~\ref{fig:ivc}.}\label{fig:stabkvot2}
\end{figure}%

In the following simulations, we have chosen $M=300$ and $\Delta t$ in
accordance with \eqref{eq:stability}--\eqref{eq:CFL}. In the next two
simulations, we choose $\cs/\cinf=0.1$, which corresponds to Case~I.iii of
Theorem~\ref{thm:case1} and means that there are two steady-state solutions;
see the red and green dots in Figure~\ref{fig:f} and the corresponding
steady-state concentrations in Figure~\ref{fig:iii}. The simulation shown in
Figure~\ref{fig:unstabup} demonstrates simultaneously the instability of the
steady state with the shorter length $\linf^{\mathrm{short}}=2.82$~mm and the
stability of the longer one with $\linf^{\mathrm{long}}=17.3$~mm by starting
very close to the smaller steady state: $l(0)=\linf^{\mathrm{short}}+1~\mu$m
and with an initial profile $c(x,0)$ very close to $c_1(x)$ shown in
Figure~\ref{fig:iii}.
\begin{figure}[tbh]
  \centering\includegraphics[width=0.4\textwidth]{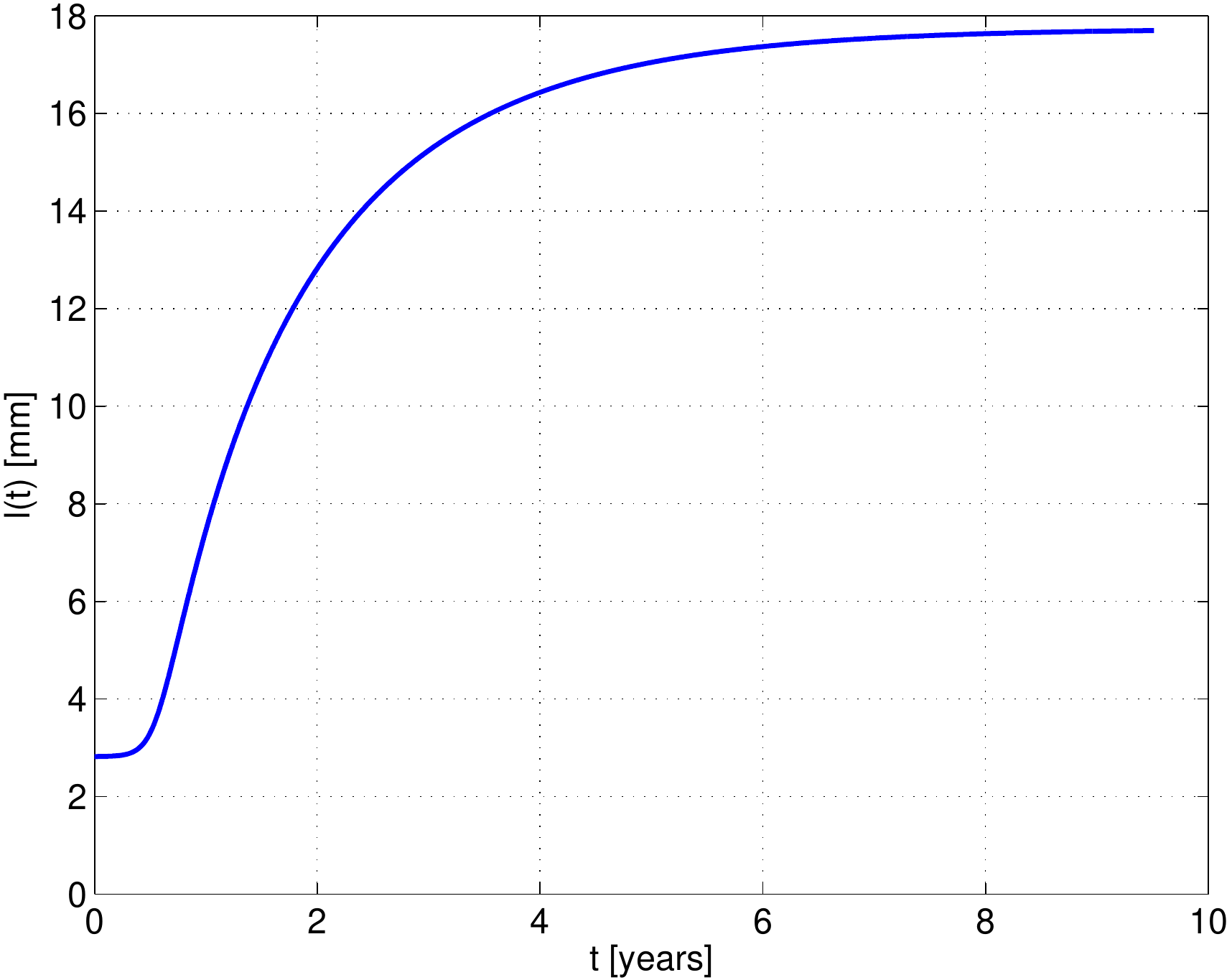}
  \caption{The case $\cs/\cinf=0.1$. Dynamic behaviour of an axon with initial length
  slightly larger than the steady-state value $\linf^{\mathrm{short}}=2.82$~mm.
   The initial data is almost in steady state, but since this
  is unstable, convergence to the other stable steady state occurs.
  }\label{fig:unstabup}
\end{figure}%

If the initial data is instead a small perturbation to the slightly shorter
length  $l(0)=\linf^{\mathrm{short}}-1~\mu$m, we get the simulation result in
Figure~\ref{fig:unstabdown}.
\begin{figure}[tbh]
  \centering\includegraphics[width=0.4\textwidth]{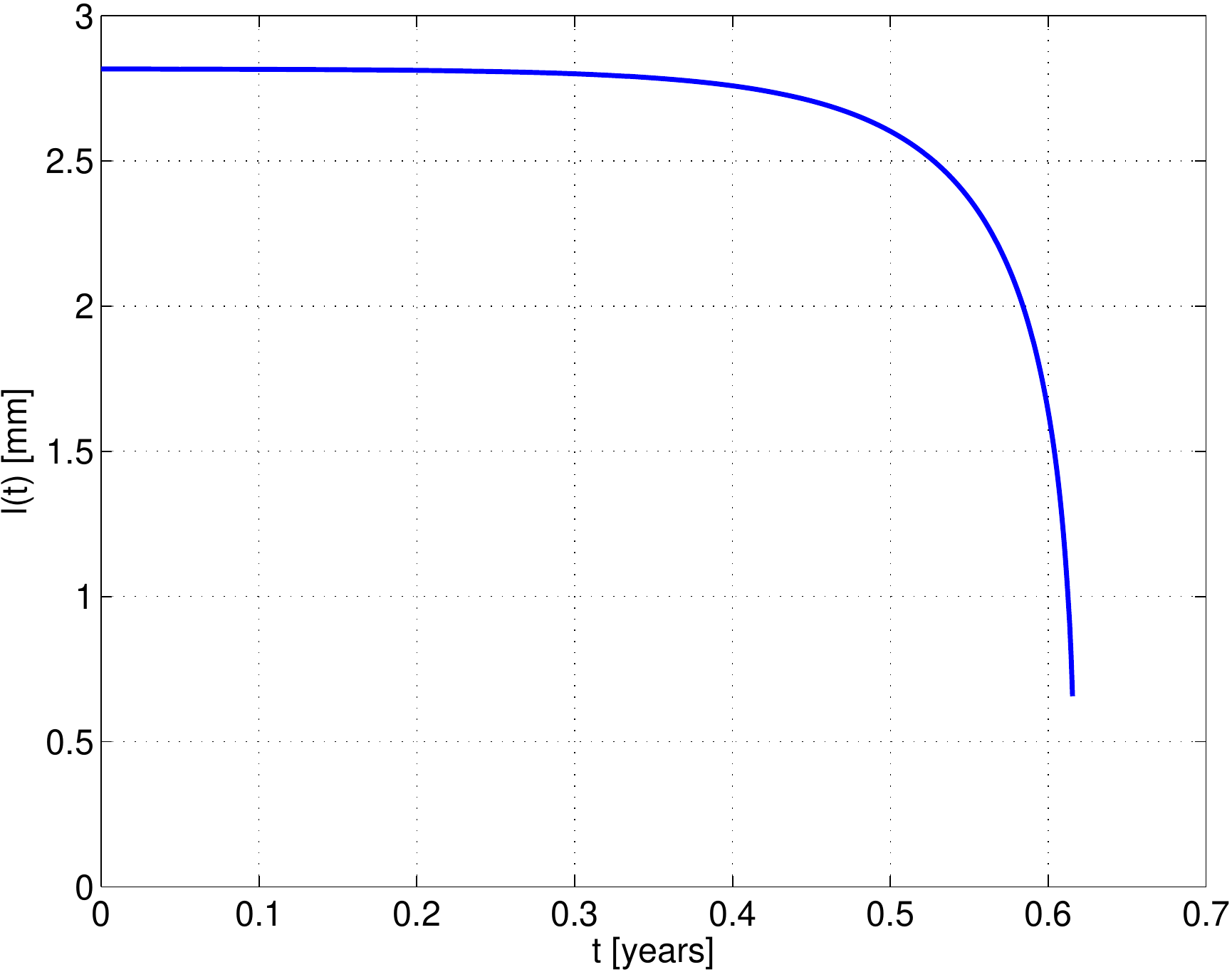}
  \caption{The case $\cs/\cinf=0.1$. The dynamic behaviour of an axon with initial length
  slightly smaller than the steady-state value $\linf^{\mathrm{short}}=2.82$~mm.
 The initial data is almost in steady state, but since this
  is unstable the axon shrinks.
  }\label{fig:unstabdown}
\end{figure}%
Note that our model \eqref{eq:model} is only valid for $l(t)>0$. In particular,
zero length of an axon is not a steady state. This is in agreement with the
simulation in Figure~\ref{fig:unstabdown}, in which $l'(t)$ is far from zero as
$l(t)\rightarrow 0$.

With the kept ratio $\cs/\cinf=0.1$, Figure~\ref{fig:stabdown} shows a
simulation when the initial length is $l(0)=30$~mm, which is greater than the
stable steady-state length $\linf^{\mathrm{long}}=17.3$~mm. There is a
convergence back to the stable length 17.3~mm.
\begin{figure}[tbh]
  \centering\includegraphics[width=0.4\textwidth]{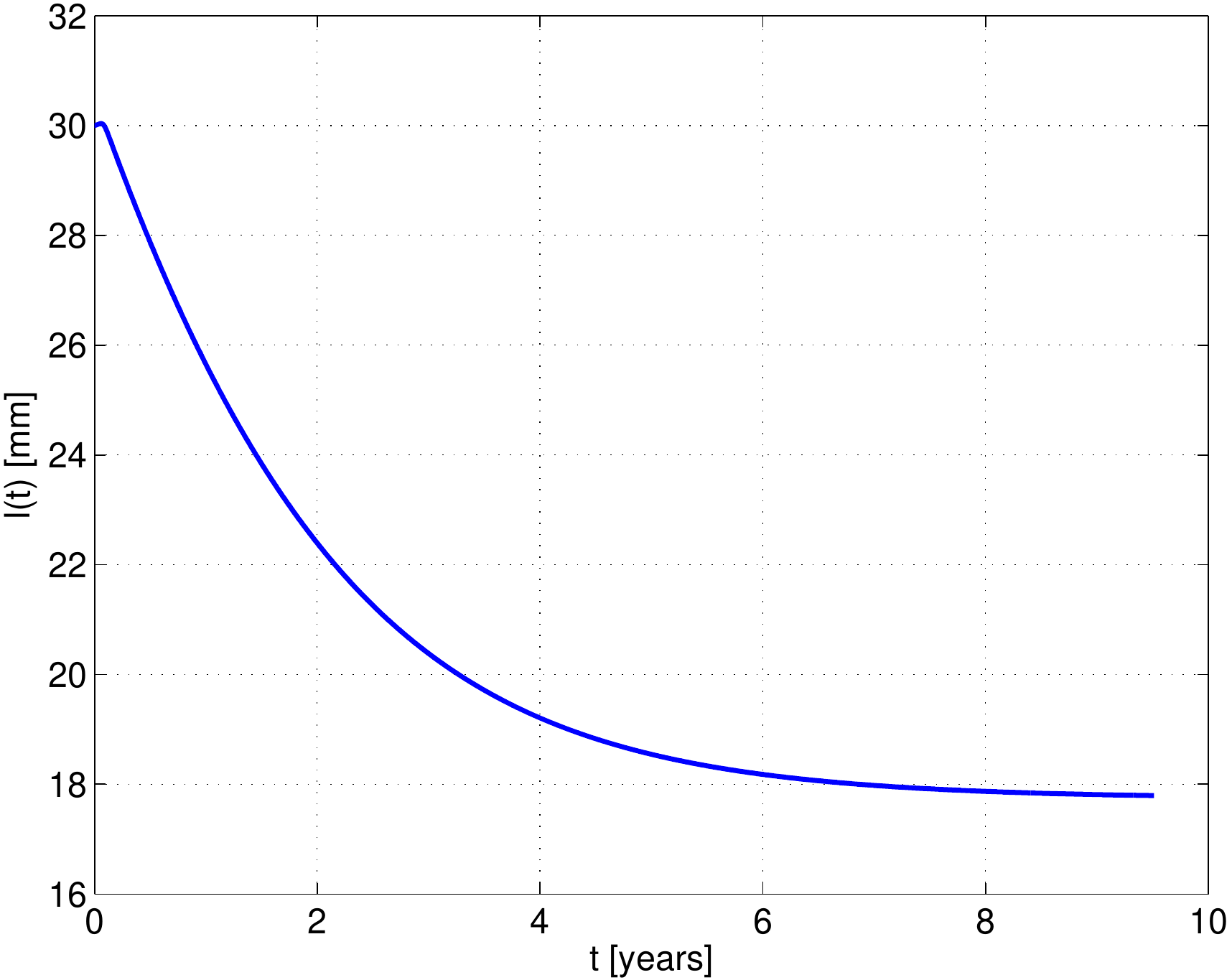}
  \caption{The case $\cs/\cinf=0.1$. Starting from the initial length of $l(0)=30$~mm
  there is a convergence to the steady-state length $\linf^{\mathrm{long}}=17.3$~mm.
  }\label{fig:stabdown}
\end{figure}%

Finally, we show a simulation when the initial length is $l(0)=30$~mm, but the
ratio has been lowered to $\cs/\cinf=0.05$, which is below the threshold value
$f(z_0)=0.060$. According to Theorem~\ref{thm:case1} there exists no steady
state, which is in accordance with the simulation in Figure~\ref{fig:noSS},
which shows how the axon shrinks.
\begin{figure}[tbh]
  \centering\includegraphics[width=0.4\textwidth]{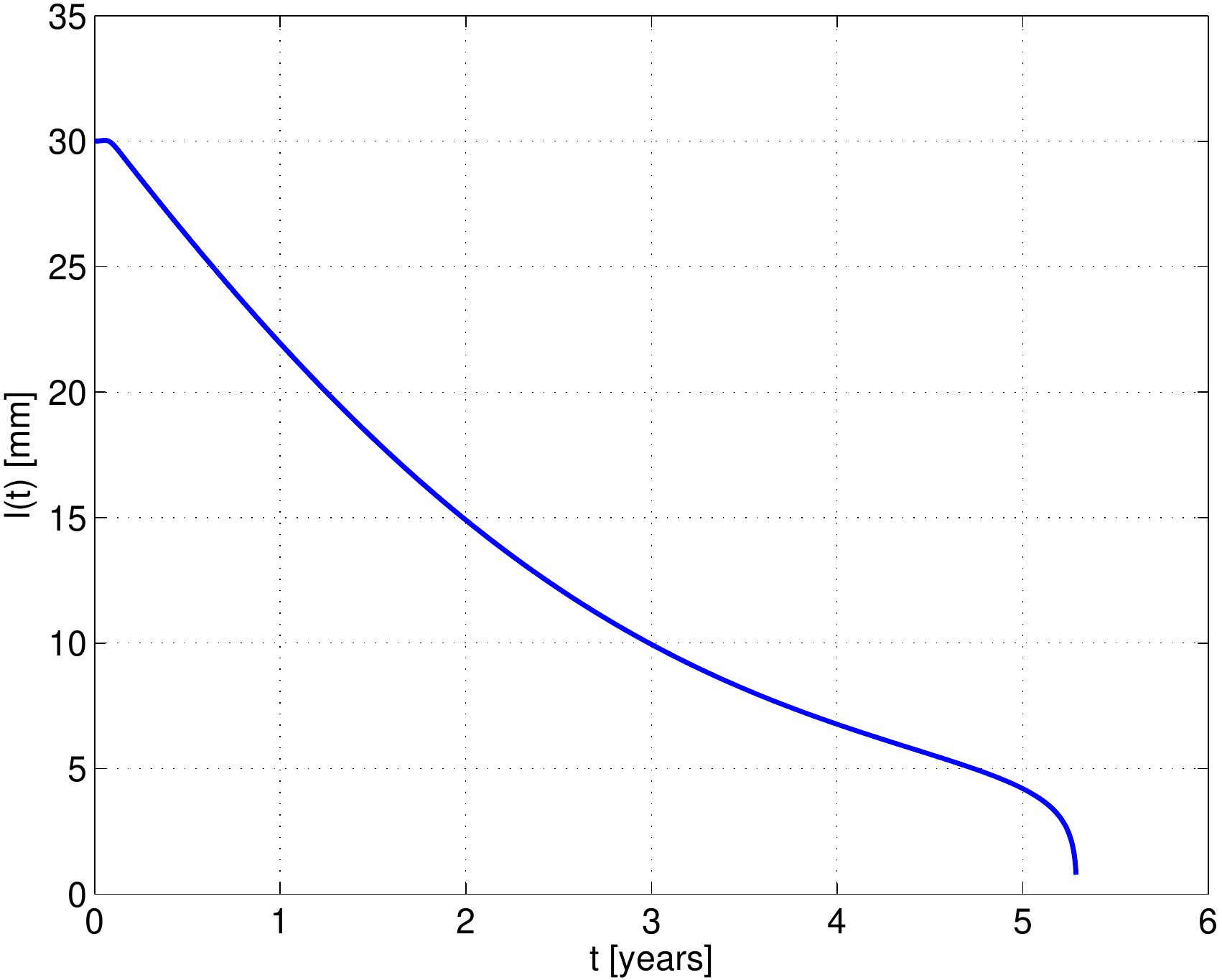}
  \caption{The case $\cs/\cinf=0.05$ when no steady state exists.
  Starting from the initial length of $l(0)=30$~mm
  the axons shrinks.
  }\label{fig:noSS}
\end{figure}%

\section{Discussion and conclusions}\label{sect:discussion}

\subsection{Mathematical modelling aspects}

The dynamic model presented for axonal growth \eqref{eq:model} can be seen as
an extension and modification of a previously published model by McLean and
Graham~\cite{McLean2004} (MG~model). Both models use the linear PDE
\eqref{eq:PDE1} for modelling the active and diffusive components of transport
of tubulin along the axon and an ODE for the axonal elongation
\eqref{eq:dldt2}, namely that the length increase per time unit is equal to an
affine relationship of the available concentration of free tubulin in the
growth cone (or at the boundary $x=l(t)$).

The two new ingredients in our model is (1) an additional ODE for the
concentration of free tubulin in the growth cone; and (2) a careful derivation
of the flux of free tubulin into the growth --- this flux includes the velocity
of the moving boundary. We allow the growth cone to have a certain size,
whereas it is zero in the MG~model. When the size of the growth cone in our
model tends to zero, the MG~model is, however, not obtained. The reason lies in
the ingredient (2). We explain this in detail below.

The MG~model has an additional ODE for the production of tubulin in the soma.
This means that the production rate and size of the soma must be known to drive
the model. This is an interesting modelling aspect that we have purposely left
out, since the boundary at $x=0$ is stationary and implies no special
difficulty in the modelling. We have focused on the complexity of the moving
boundary at $x=l(t)$ and have chosen the simpler assumption that the soma
concentration is directly given as a function of time; $\cs=\cs(t)$.

In the modelling step, all the following phenomena or processes have been taken
into account: active transport, diffusion, degradation along the axon and in
the growth cone, assembly and disassembly of tubulin dimers in the growth cone
and finally the movement of the growth cone, which is also the output of the
model. The problematic fact is that all these effects are combined and
influence each other, since they depend on the local concentration. With such
many combined processes, it is impossible to tell beforehand (by physiological
and biological experience) what a natural steady-state concentration
distribution $c(x)$ along the axon is. It is then the strength of mathematical
modelling comes in. Each phenomenon can be included individually, based on a
restricted physiological or biological explanation possibly obtained from
controlled experimental conditions; see Section~\ref{sect:parameters} and all
references therein. The outcome of the mathematical model is then the
combination of all phenomena such that they obey the overall physical law of
conservation of mass. A typical such example is the non-monotone concentration
distribution of free tubulin $c(x)$ along the axon, which our model yields for
nominal parameter values. Having a correct intuitive feeling for the final
outcome would be equal to solving the differential equations by intuition. On
the other hand, once we have the mathematical result (the non-monotone
profile), one should try to give some intuitive explanations; see
Section~\ref{sect:conclusions}.

It is interesting to track down the exact details of the differences between
the present model and the MG~model. We use the same PDE~\eqref{eq:PDE1} along
the axon, which comes from the conservation of mass and the fact that the flux
[mol/s] of tubulin at any fixed point on the $x$-axis is $ac^--Dc_x^-$. We also
use the same ODE for the axonal elongation \eqref{eq:dldt2}. The second ODE in
our model describes the accumulation of free tubulin in the growth cone and
models thereby the dynamic variation of the cone concentration $\cc(t)$. In
this way, we get a natural connection between the transport processes along the
axon and the growth cone concentration.

The MG~model has a growth cone of size zero, which is a mathematical
idealization that can be made. However, the main difference lies in the
boundary condition at $x=l(t)$. We have carefully derived that the flux of
tubulin over $x=l(t)$ (seen by an observer moving with the boundary) is given
by the expression \eqref{eq:fluxoverl}, i.e., the flux per area unit
[mol/(m$^2$s)] is $ac^--Dc_x^--l'(t)c^-$. McLean and Graham~\cite{McLean2004}
writes that this flux is $-c_x^-$. Hence, they assume that diffusion is the
only effective flux (with $D=1$) and that $a=l'(t)$. (The latter equality
implies $l(t)=at+l_0$, which contradict the ODE for $l(t)$.) In any case, their
final boundary condition \cite[Equation~(2.3)]{McLean2004} for $x=l(t)$ is
\begin{align}\label{eq:McLean}
&c_x^-=-\epsilon_lc^-+\zeta_l.
\end{align}
This introduction of new parameters remove the above implicit assumption that
$D=1$. However, \eqref{eq:McLean} means that the moving boundary has not been
taken into account. To compare with our model, Equation~\eqref{eq:McLean} can
be seen as a simplification of our ODE for the growth cone concentration, which
is either \eqref{eq:cons_law_cone}, \eqref{eq:dccdt1} or the second equation of
\eqref{eq:model}. Namely, setting the size parameter $\lc$ of the growth cone
to zero and using $\cc=c^-$, \eqref{eq:dccdt1} becomes
\begin{gather}
0=ac^--Dc_x^--\rg(c^--\cinf)c^-\quad\Longleftrightarrow\quad\notag\\
c_x^-=-\frac{\rg}{D}(c^-)^2+\frac{a+\rg\cinf}{D}c^-.\label{eq:temp1}
\end{gather}
An equivalent way of obtaining this equation is to simply use the conservation
of mass over $x=l(t)$. Seen from an observer moving with the boundary (to the
right along the $x$-axis), the influx on the left-hand side of $x=l(t)$ is
$ac^--Dc_x^--l'(t)c^-$, which is equal to the flux out on the right-hand side,
which is zero. Substituting the ODE for $l'(t)$ one arrives at
\eqref{eq:temp1}.

The differences between the two models (when the growth cone has zero size) can
now be seen in the differences between \eqref{eq:McLean} and \eqref{eq:temp1}.
Firstly, in \eqref{eq:temp1} no new parameter has been introduced, in contrast
to $\epsilon_l$ and $\zeta_l$ \eqref{eq:McLean}. Secondly, taking the moving
boundary into account, one ends up in a nonlinear relation \eqref{eq:temp1}
between the concentration $c^-$ and its spatial derivative $c_x^-$. This
difference explains why the MG~model yields decreasing concentration profiles
along the axon, whereas we get non-monotone ones.

Note that we get a non-monotone profile also in the case when $\lc=0$;
Theorem~\ref{thm:case1} still holds; see also Figure~\ref{fig:varyg}. Then
there is no ODE for the growth cone, but instead the boundary condition
\eqref{eq:temp1} at $x=l(t)$ for the PDE. Hence, it is not the growth cone ODE
that is responsible for the non-monotone steady-state profile, it is simply the
outcome of the conservation law of mass.

\subsection{Conclusions}\label{sect:conclusions}

One outcome of this work is the presented new dynamic model of tubulin-driven
axonal growth \eqref{eq:model}. The input to the model is the tubulin
concentration $\cs(t)$, which we assume is given. Except for this function, the
model has seven (constant) parameters, which are all biological or
physiological parameters or combinations of such. Since we analyze steady-state
solutions in this work, $\cs$ has been constant in time. Depending on this
value, and the values of five additional parameters ($a$, $D$, $g$, $\lc$,
$\cinf$), all steady-state solutions have been classified. One conclusion is
that the values of the rate constants $\rg$ and $\rgtilde$ in the dynamic model
have no influence on the steady states, hence only on the dynamic behaviour and
convergence to steady state.

The biologically most interesting cases arise when all parameters take positive
values and these cases are given in Theorem~\ref{thm:case1}. With the nominal
parameter values (Table~\ref{table:1}), which we have extracted directly or
indirectly from the biological literature, the following inequality is
fulfilled:
$$
g\lc=2\cdot 10^{-12}\text{~m$/$s}~<1\cdot
10^{-8}\text{~m$/$s~}=a.
$$
Then Case~I of Theorem~\ref{thm:case1} states that, when a steady state exists,
the concentration of tubulin along the axon $c(x)$ is given by the explicit
formula \eqref{eq:SSsolution}, but the length $\linf$ has to be obtained
numerically by solving the equation $f(z)=\cs/\cinf$ with $f(z)$ given by
\eqref{eq:f}. A way of getting an overview of the possible values of $\linf$ is
to plot the graph of the function \eqref{eq:f}, which has a minimum at
$z_0=5.64$~mm; see Figure~\ref{fig:f}. The value on the ratio $\cs/\cinf$ gives
directly whether zero, one or two steady states exist and the corresponding
length(s) $\linf$ can be read off on the $z$-axis. If the ratio
$\cs/\cinf<f(z_0)=0.06$, then there exists no steady state, i.e.\ no outgrowth
can occur. For higher values on $\cs/\cinf$, the steady-state concentration
along the axon is given by the function \eqref{eq:SSsolution}. If
$0.06<\cs/\cinf<1$, then there exists two steady-state solutions. One has a
very short $\linf<z_0=5.64$~mm, which actually becomes smaller and smaller the
closer the ratio $\cs/\cinf$ is to 1. Dynamic numerical simulations indicate
that this is an unstable steady state and thus not possible biologically. The
other steady state is stable, has a larger $\linf>z_0=5.64$~mm, and this value
is larger the larger $\cs/\cinf$ is. When $\cs/\cinf\ge 1$, there exists a
unique steady-state solution.

Numerical simulations of the dynamic behaviour in Section~\ref{sect:stability}
indicate that every steady-state solution with $\linf>z_0=5.64$~mm is stable,
whereas the shorter ones are unstable. An overall conclusions is thus that for
the most interesting cases when $g\lc<a$ and $\cs/\cinf>f(z_0)$ hold, there
exists a unique stable steady-state solution. For an axon of arbitrary length,
if the soma concentration decreases to a constant value such that
$\cs/\cinf<f(z_0)$, then there exists no steady-state solution and dynamic
numerical simulations show a shrinking axon.

One interesting outcome from the model is the form of the concentration
distribution $c(x)$ along the axon for the stable steady states in Case~I of
Theorem~\ref{thm:case1}. Starting from the soma at $x=0$, the concentration
decreases and is in fact along a large portion of the axon lower than both
$\cs$ and the steady-state soma concentration $\cinf$. The explanation for this
form is hard to make intuitively, since it is a result of the combined effects
of active transport, diffusion, degradation, (dis)assembly of tubulin and the
velocity of the growth cone. The mathematical equations combine these phenomena
and yield the very form of $c(x)$, which is given explicitly by the closed-form
expression \eqref{eq:SSsolution}. This is a good example of the purpose and
strength of mathematical modelling; it is impossible to reason biologically
what the precise form of a steady-state solution should be.

An interesting biological conclusion, which is in agreement with experimental
results, is that a relatively large active transport velocity $a$ means that
the \emph{flux} is sufficiently large to transport tubulin the long way out to
the growth cone despite the relatively low \emph{concentration} of tubulin in
the axon. Along the axon, the two flux components, advective and diffusive, are
precisely balanced by the degradation of tubulin. Note that the total flux
(advective plus diffusive) is decreasing along the axon because of the
degradation. The increasing concentration distribution near the growth cone
implies that diffusion occurs in the direction to the soma but the active
transport to the growth cone is so large that the net flux is precisely what is
needed to balance the degradation of tubulin in the growth cone. Note also that
we have assumed that the concentration of tubulin varies continuously due to
the presence of diffusion.

If the active transport velocity $a$ is too low, namely less than $g\lc$
(Case~II of Theorem~\ref{thm:case1}), then there still exists a unique
solution; however, the concentration along the axon $c(x)$ is decreasing, which
partly means that the diffusion flux is always directed towards to growth cone,
and partly that this solution only exists if the soma concentration
$\cs>\cinf$.

The model gives information also in the extreme cases when one or more
parameters are zero. The biological interest in such cases is if the
corresponding variable(s) are negligible. The mathematical advantage to set a
very small variable to zero, is that in some cases one gets special solutions
that can be written up explicitly which makes it easier to draw biological
conclusions.

When advection is negligible ($a=0$), the flux of tubulin from the soma to the
growth cone is only present in the form of diffusion. Since diffusive flux
occurs from higher to lower concentrations, the only possible steady state has
a decreasing concentration of tubulin from the soma to the growth cone, which
is precisely what Theorem~\ref{thm:case2} states. This theorem also states that
if the soma concentration is too small ($\cs<\cinf$), there exists no steady
state with $\linf>0$.

Another extreme case is when the degradation of tubulin is negligible ($g=0$).
Theorem~\ref{thm:case3} states that a steady state exists if and only if
$\cs<\cinf$ holds. Then there is no net flux at any $x$ along the axon; i.e.\
the concentration $c(x)$ is increasing in such away that the advective flux
towards the growth cone is equal to the diffusive flux back to the soma. If
$\cs\ge\cinf$, then there exists no steady state and the the axon may grow
indefinitely.

The case when diffusion is negligible ($D=0$) is a special case; however, not
an unrealistic case, since $D$ is a relatively small number. Without diffusion,
the solution of the PDE may in fact contain discontinuities. Hence, the
continuity assumption we have made for the boundary conditions is not natural.
From a mathematical point of view, one cannot beforehand ignore the possibility
of a concentration discontinuity at $x=l(t)$. Theorem~\ref{thm:SSdisc} gives
explicit functions for both the decreasing concentration profile along the axon
and the length of the axon in the unique steady state. Note that the decreasing
concentration along the axon agrees with the general case except near the
growth cone. We also note that the steady state exists if and only if
$a\cs>g\lc\cinf$, i.e., when the advective flux ($a\cs$) from the soma is
greater than the degradation of tubulin in the growth cone ($g\lc\cinf$). The
difference between these two numbers is precisely the amount per time unit of
transported tubulin along the axon that is degraded.

We have not found (in literature) any experimental indications against the
non-monotone concentration distribution of tubulin along the axon which our
model yields for the nominal parameter values. It is in agreement with the fact
that the active transport is the most important ingredient for the outgrowth of
long axons. The velocity $a$ of the active transport has been assumed to be
constant. Experiments reported by Watson et al.~\cite{Watson1989} and Xu and
Tung~\cite{Xu2001} show a decreasing velocity along fully grown axons; i.e.
$a=a(x)$ is a decreasing function. Although such a dependence could be included
in the model, it is not clear what causes this decrease and hence what function
$a(x)$ to use.

More comprehensive future extensions of the model would be to include other
substances than tubulin, for example, actin which is redistributed within the
growth cone so that it can turn as a response to external stimuli.

\section*{Acknowledgement}

This article is dedicated to Professor Martin Kanje, Department of Biology,
Lund University, who initiated this work, but unfortunately deceased on March
21, 2013, before he was able to see and interact with the results. Erik
Henningsson was supported by the Swedish Research Council Grant no.\
621-2011-5588. We would like to thank the anonymous reviewers for critical
questions, which lead to a substantially improved article.





\bibliographystyle{plain}


\end{document}